\documentclass[a4paper,10pt,numbers=noenddot,twoside=on,headinclude=true,footinclude=false,headsepline=true]{scrartcl}

\usepackage[utf8]{inputenc}
\usepackage[T1]{fontenc}
\usepackage{textcomp}
\usepackage[english]{babel}
\usepackage{color}
\usepackage{amssymb}
\usepackage{amsmath}
\usepackage{amsfonts}
\usepackage{amsopn}
\usepackage{amsbsy}
\usepackage[overload,ntheorem]{empheq}
\usepackage{enumerate}
\usepackage{graphics}
\usepackage{units}
\usepackage[cal=boondoxo,bb=boondox]{mathalfa}
\usepackage{dsfont}

\numberwithin{equation}{section} 
\numberwithin{table}{section} 
\numberwithin{figure}{section}

\usepackage[amsmath,thmmarks,thref]{ntheorem}

\theoremstyle{plain}
\newtheorem{theorem}{Theorem}[section]

\newtheorem{conjecture}[theorem]{Conjecture}
\newtheorem{lemma}[theorem]{Lemma}
\newtheorem{corollary}[theorem]{Corollary}
\newtheorem{proposition}[theorem]{Proposition}
\newtheorem{assumption}[theorem]{Assumption}

\theorembodyfont{\upshape}

\theoremstyle{nonumberplain}

\theoremsymbol{\ensuremath{_\Box}}
\newtheorem{proof}{Proof}

\usepackage[style=alphabetic,
	citestyle=alphabetic,
	firstinits=true,
	backend=biber,
	hyperref=auto,
	sorting=nyt,
	maxcitenames=3,
	maxbibnames=99,
	minnames=3,
	arxiv=pdf,
	url=false,
	isbn=false,
	dateabbrev=true,
	datezeros=true,
	bibencoding=utf8]{biblatex}
\newbibmacro{string+doi}[1]{%
	\iffieldundef{doi}{#1}{\href{http://dx.doi.org/\thefield{doi}}{#1}}}
\DeclareFieldFormat
	{title}{\usebibmacro{string+doi}{\mkbibemph{#1}}\isdot}
\DeclareFieldFormat
	[article,inbook,incollection,inproceedings,patent,thesis,unpublished]
	{title}{\usebibmacro{string+doi}{\mkbibemph{#1}}\isdot}
\DeclareFieldFormat
	[article,inbook,incollection,inproceedings,patent,thesis,unpublished]
	{pages}{#1}
\DeclareFieldFormat
	[article]
	{journaltitle}{#1\isdot}
\DeclareFieldFormat
	[article]
	{volume}{\textbf{#1}}
\DefineBibliographyStrings{english}{pages={}}
\DeclareBibliographyDriver{article}{%
	\usebibmacro{bibindex}%
	\usebibmacro{begentry}%
	\usebibmacro{author/translator+others}%
	\setunit{\labelnamepunct}\newblock
	\usebibmacro{title}%
	\usebibmacro{journal+issuetitle}%
	\setunit*{\addcomma\space}
	\printfield{pages}
	\setunit*{\addcomma\space}
	\printfield{year}
	\usebibmacro{finentry}%
}

\newbibmacro*{journal+issuetitle}{%
	\usebibmacro{journal}%
	\setunit*{\addspace}%
	\iffieldundef{series}
	{}
	{\newunit
		\printfield{series}%
		\setunit{\addspace}}%
	\iffieldundef{volume}
		{}
		{\printfield{volume}%
		\setunit{\addspace}}%
	\setunit*{\addcolon\space}%
	\usebibmacro{issue}%
	\newunit}

\usepackage[scaled]{beramono}
\usepackage{helvet}
\usepackage[charter]{mathdesign} 
\SetMathAlphabet{\mathcal}{normal}{OMS}{cmsy}{m}{n} 
\SetMathAlphabet{\mathcal}{bold}{OMS}{cmsy}{m}{n} 
\providecommand{\ie}{i.~e.~}
\providecommand{\eg}{e.~g.~}
\providecommand{\cf}{cf.~}

\providecommand{\R}{\mathbb{R}}

\providecommand{\C}{\mathbb{C}}
\renewcommand{\C}{\mathbb{C}}

\providecommand{\T}{\mathbb{T}}
\renewcommand{\T}{\mathbb{T}}

\providecommand{\Z}{\mathbb{Z}}
\providecommand{\ii}{\mathrm{i}}
\providecommand{\e}{\mathrm{e}}
\renewcommand{\Re}{\mathrm{Re} \,}
\renewcommand{\Im}{\mathrm{Im} \,}
\providecommand{\Hil}{\mathcal{H}}

\providecommand{\eps}{\varepsilon}

\providecommand{\ran}{\mathrm{ran} \, }

\providecommand{\dd}{\mathrm{d}}
\providecommand{\id}{\mathds{1}}

\providecommand{\Fourier}{\mathcal{F}}

\providecommand{\scpro}[2]{\left \langle #1 , #2 \right \rangle}
\providecommand{\sscpro}[2]{\langle #1 , #2 \rangle}
\providecommand{\bscpro}[2]{\bigl \langle #1 , #2 \bigr \rangle}

\providecommand{\sopro}[2]{\vert #1 \rangle \langle #2 \vert}

\providecommand{\Rot}{\mathrm{Rot} \,}
\providecommand{\Div}{\mathrm{Div} \,}

\providecommand{\Jphys}{\mathcal{J}_+}
\providecommand{\Maux}{M^{\mathrm{aux}}}
\providecommand{\BZ}{\mathbb{B}}

\providecommand{\specrel}{\sigma_{\mathrm{rel}}}

\usepackage{slashbox}
\usepackage{caption}
\usepackage{subcaption}

\usepackage{tabularx}

\bibliography{/Users/max/Library/texmf/tex/latex/max/bibliography}

\title{Symmetry Classification of \\ Topological Photonic Crystals}
\author{Giuseppe De Nittis${}^1$ \& Max Lein${}^2$}

\begin{document}

\maketitle
\vspace{-9mm}
\begin{center}
	${}^1$ Facultad de Matemáticas \& Instituto de Física, 
	Pontificia Universidad Católica de Chile \linebreak
	Avenida Vicuña Mackenna 4860, 
	Santiago, 
	Chile \linebreak
	{\footnotesize \href{mailto:gidenittis@mat.uc.cl}{\texttt{gidenittis@mat.uc.cl}}}
	\medskip
	\\
	${}^2$ Advanced Institute of Materials Research,  
	Tohoku University \linebreak
	2-1-1 Katahira, Aoba-ku, 
	Sendai, 980-8577, 
	Japan \linebreak
	{\footnotesize \href{mailto:maximilian.lein.d2@tohoku.ac.jp}{\texttt{maximilian.lein.d2@tohoku.ac.jp}}}
\end{center}
\begin{abstract}
	In a seminal paper Haldane conjectured that topological phenomena are not particular to quantum systems, and indeed experiments realized unidirectional, backscattering-free edge modes with \emph{electromagnetic waves}. This raises two immediate questions: (1)~Are there other topological effects in electromagnetic media? And (2)~is Haldane's “Quantum Hall Effect for light” really analogous to the Quantum Hall Effect? 
	
	We conclusively answer both of these questions by \emph{classifying topological photonic crystals} according to material (as opposed to crystallographic) symmetries. It turns out there are \emph{four topologically distinct types of media}, of which only one, gyrotropic media, is topologically non-trivial in $d = 2 , 3$. That means there are no as-of-yet undiscovered topological effects; in particular, there is no analog of the Quantum \emph{Spin} Hall Effect in classical electromagnetism. Moreover, at least \emph{qualitatively}, Haldane's Quantum Hall Effect for light is analogous to the Quantum Hall Effect from condensed matter physics as both systems as in the same topological class, class~A. Our ideas are directly applicable to other classical waves. 
\end{abstract}
\\
{\scriptsize \textbf{PACS 2010:} 41.20.Jb, 42.70.Qs, 78.20.-e}

\newpage
\tableofcontents

\section{Introduction} 
\label{intro}
Raghu and Haldane proposed in a seminal work \cite{Raghu_Haldane:quantum_Hall_effect_photonic_crystals:2008} that topological effects are \emph{bona fide wave} rather than \emph{quantum} phenomena. In analogy to the Quantum Hall Effect, they predicted unidirectional, backscattering-free edge modes in periodic gyrotropic light conductors, also known as gyrotropic \emph{photonic crystals}, and attributed their existence to the “non-trivial topology of the system”. More specifically, they made the following 
\begin{conjecture}[Raghu and Haldane's Photonic Bulk-Edge Correspondence \cite{Raghu_Haldane:quantum_Hall_effect_photonic_crystals:2008}]\label{intro:conjecture:photonic_bulk_edge_correspondence}
	In a two-dimensional photonic crystals with boundary the difference of the number of left- and right-moving boundary modes in bulk band gaps is a topologically protected quantity and equals the Chern number associated to the frequency bands below the bulk band gap. 
\end{conjecture}
Raghu and Haldane base their arguments on the analogy to the corresponding quantum systems: they \emph{proposed} (as opposed to derived) ray optics equations, which contain geometric Berry curvature terms to subleading order. And since the Chern number can be computed as the Brillouin zone average of the Berry curvature, the analogy to the Bloch electron is then invoked in an \emph{ad hoc} fashion without making any reference to the underlying dynamical equations. Their prediction has been confirmed in a number of spectacular experiments in electromagnetic, acoustic and phononic waves \cite{Wang_et_al:unidirectional_backscattering_photonic_crystal:2009,Rechtsman_Zeuner_et_al:photonic_topological_insulators:2013,Lu_et_al:experimental_observation_Weyl_points:2015,Xiao_et_al:geometric_phase_acoustic_systems:2015,Suesstrunk_Huber:mechanical_topological_insulator:2015}. Up until now a first-principles understanding starting from Maxwell's equations is an open problem. These and other, more recent works have naturally raised two questions: 
\begin{enumerate}[(1)]
	\item How similar is the Quantum Hall Effect for light to the one from solid state physics? 
	\item Are there other, as-of-yet unknown topological effects in electromagnetic media? 
\end{enumerate}
To answer these questions and get a more complete picture, we will rigorously establish what “topological” means in the context of classical electromagnetism. This keyword is inserted into the discussion of a lot of physical effects — even if it is not always clear what that actually means. For example, three different groups \cite{Khanikaev_et_al:photonic_topological_insulators:2013,Bliokh_Smirnova_Nori:spin_orbit_light:2015,Wu_Hu:topological_photonic_crystal_with_time_reversal_symmetry:2015} have claimed to have found a photonic analog of the Quantum \emph{Spin} Hall Effect, implying that the spin-momentum locking they find is of topological origin. Our first principles approach will clear up this confusion, and we will analyze these three works in the conclusion (Section~\ref{conclusion:comparison:QSHE}). 

The crucial ingredient in the analysis of topological effects is \emph{symmetries}, and when designing topological electromagnetic media, there are two axes to explore: One can choose the materials from which to build the photonic crystal (\emph{material symmetries}) and then decide how to periodically arrange these materials (\emph{crystallographic symmetries}). In this work we will focus on material symmetries. For those we answer both of these questions conclusively by first reformulating Maxwell's equations in Schrödinger form \cite{DeNittis_Lein:Schroedinger_formalism_classical_waves:2017}, and then adapting the Cartan-Altland-Zirnbauer classification scheme for topological insulators. The latter is the content of the present paper and the fifth in a sequence of earlier works \cite{DeNittis_Lein:adiabatic_periodic_Maxwell_PsiDO:2013,DeNittis_Lein:sapt_photonic_crystals:2013,DeNittis_Lein:symmetries_Maxwell:2014,DeNittis_Lein:ray_optics_photonic_crystals:2014} that tries to systematically understand how topological effects emerge from electrodynamics.

Initially, the term topological insulator was born of the \emph{topological interpretation of the Quantum Hall Effect}  \cite{Thouless_Kohmoto_Nightingale_Den_Nijs:quantized_hall_conductance:1982}. This seminal work by Thouless, Kohmoto, Nightingale and den Nijs linked a \emph{measurable} quantity, the transverse conductivity, to a \emph{topological invariant} of the so-called Bloch vector bundle, the Chern number. Put succinctly, they have established a connection between the topology of a mathematical object, in this case a vector bundle, and a concrete physical quantity, the transverse conductivity. Topological insulators come in more than one flavor, and their systematic classification and characterization (see \eg \cite{Hasan_Kane:topological_insulators:2010,Chiu_Teo_Schnyder_Ryu:classification_topological_insulators:2016,Prodan_Schulz_Baldes:complex_topological_insulators:2016} for three recent reviews) is the main aim of several vibrant communities within theoretical and mathematical physics; the most common classification tool is the so-called Cartan-Altland-Zirnbauer (CAZ) scheme \cite{Altland_Zirnbauer:superconductors_symmetries:1997,Schnyder_Ryu_Furusaki_Ludwig:classification_topological_insulators:2008,Chiu_Teo_Schnyder_Ryu:classification_topological_insulators:2016}. So not only can \emph{breaking} symmetries can lead to topological effects, also their \emph{presence} might \cite{Kane_Mele:Z2_ordering_spin_quantum_Hall_effect:2005,Prodan_Schulz_Baldes:complex_topological_insulators:2016,DeNittis_Gomi:AII_bundles:2014,DeNittis_Gomi:AIII_bundles:2015}. The idea to realize an analog of the Quantum Hall Effect with electromagnetic waves in photonic crystals relies on the \emph{breaking of an even time-reversal symmetry}. 

Just like with (quantum) topological insulators the existence of “conducting” electromagnetic boundary states in a region where the bulk is “insulating” is expected to be explainable via \emph{bulk-boundary correspondences}; in the context of classical waves conducting means the presence of states while insulating refers to their absence. This idea goes back to Hatsugai's works on topological condensed matter systems \cite{Hatsugai:edge_states_Riemann_surface:1993,Hatsugai:Chern_number_edge_states:1993} and states that certain aspects of the system at the boundary are completely determined by its properties in the interior. Bulk-boundary, sometimes also known as bulk-edge correspondences, usually consist of three equalities: 
\begin{subequations}\label{intro:eqn:bulk_edge_correspondence}
	\begin{align}
		O_{\mathrm{bulk}/\mathrm{edge}}(t) &\approx T_{\mathrm{bulk}/\mathrm{edge}} 
		\label{intro:eqn:bulk_edge_correspondence:topological_observable}
		\\
		T_{\mathrm{bulk}} &= T_{\mathrm{edge}} 
		\label{intro:eqn:bulk_edge_correspondence:mathematical_bulk_edge}
	\end{align}
\end{subequations}
The first two identify physical observables $O_{\mathrm{bulk}/\mathrm{egde}}$ in the bulk and at the boundary that are approximately given in terms of topological quantities $T_{\mathrm{bulk}/\mathrm{edge}}$, and the third states that the two topological quantities necessarily agree with one another. The number and nature of topological invariants depends on the presence or absence of certain discrete symmetries. For the Quantum Hall Effect an even time-reversal symmetry is broken by the magnetic field, and \cite{Thouless_Kohmoto_Nightingale_Den_Nijs:quantized_hall_conductance:1982} proved equation~\eqref{intro:eqn:bulk_edge_correspondence:topological_observable} in the bulk while Hatsugai contributed the other two. Note that in the more mathematically minded subcommunity, usually \eqref{intro:eqn:bulk_edge_correspondence:mathematical_bulk_edge} by itself is referred to as bulk-edge correspondence, but we insist that it is \eqref{intro:eqn:bulk_edge_correspondence:topological_observable} which imbues $T_{\mathrm{bulk}} = T_{\mathrm{edge}}$ with physical meaning. 

Therefore, justifying the Quantum Hall Effect for Light from first principles rests on proving \emph{photonic bulk-boundary correspondences}, and with this paper we aim to work towards this goal. A necessary prerequisite is to first classify linear, lossless and dispersionless media for electromagnetic waves according to certain discrete symmetries in the spirit of the CAZ scheme, as the symmetries which are present determine the \emph{number and nature of topological invariants} which are supported in photonic crystals. 

This paper provides an \emph{exhaustive classification} of electromagnetic lossless, positive index media according to their \emph{material} symmetries. Here, we will only consider symmetries which relate electric and magnetic components (as opposed to \emph{crystallographic} symmetries), \ie those of the form 
\begin{subequations}\label{intro:eqn:potential_symmetries}
	\begin{align}
		U_n &= \sigma_n \otimes \id 
		, 
		&&
		n = 1 , 2 , 3
		,
		\label{intro:eqn:potential_symmetries:linear}
		\\
		T_n &= \bigl ( \sigma_n \otimes \id \bigr ) \, C 
		, 
		&&
		n = 0 , 1 , 2 , 3
		,
		\label{intro:eqn:potential_symmetries:antilinear}
	\end{align}
\end{subequations}
where $C$ is complex conjugation, $\sigma_0 = \id$ is the identity and the $\sigma_n$ for $n = 1 , 2 , 3$ are the Pauli matrices, seen as acting on the electromagnetic splitting. For instance, $T_3 : (\mathbf{E},\mathbf{H}) \mapsto \bigl ( \overline{\mathbf{E}} , - \overline{\mathbf{H}} \bigr )$ complex conjugates the fields and garnishes $\mathbf{H}$ with a minus sign. \emph{Note that since we can always rescale the fields by constants such as $\eps_0$ and $\mu_0$, conditions such as $\eps(x) = \mu(x)$ are equivalent to $\eps(x) = \mathrm{const.} \; \mu(x)$.} To simplify presentation, we shall always assume from hereon that the fields have been suitably rescaled so that the symmetries are of the form~\eqref{intro:eqn:potential_symmetries}. 

It turns out that of those 7 symmetries only three are admissible (\cf Section~\ref{symmetries:relevant} for a detailed explanation). 
\begin{proposition}[Admissible material symmetries]
	Suppose the medium satisfies Assumption~\ref{Schroedinger:assumption:material_weights} is lossless and is a positive index material. Then of all symmetries of the form~\eqref{intro:eqn:potential_symmetries} only the \emph{even} time-reversal symmetries $T_1$ and $T_3$ as well as the dual symmetry $U_2$ are admissible. 
\end{proposition}
The presence or absence of these symmetries translate to conditions on the electric permittivity $\eps$, the magnetic permeability $\mu$ and the bianistropic tensor $\chi$ which enter the \emph{constitutive relations}
\begin{align}
	\left (
	\begin{matrix}
		\mathbf{D}(x) \\
		\mathbf{B}(x) \\
	\end{matrix}
	\right ) = W(x) \left (
	\begin{matrix}
		\mathbf{E}(x) \\
		\mathbf{H}(x) \\
	\end{matrix}
	\right ) 
	= \left (
	\begin{matrix}
		\eps(x) & \chi(x) \\
		\chi(x)^* & \mu(x) \\
	\end{matrix}
	\right ) \left (
	\begin{matrix}
		\mathbf{E}(x) \\
		\mathbf{H}(x) \\
	\end{matrix}
	\right )
	\label{intro:eqn:material_weights}
\end{align}
that relate the auxiliary fields $(\mathbf{D},\mathbf{B})$ to the electromagnetic field $(\mathbf{E},\mathbf{H})$; these enter as $3 \times 3$ blocks into the $6 \times 6$ matrix-valued function $W(x)$ which we will collectively refer to as the \emph{material weights}. $W$ phenomenologically describes how the microscopic charges in the medium react to impinging electromagnetic waves. 

Just like in quantum mechanics, for the purpose of a topological classification, only the anti-linear symmetries $T_1$ and $T_3$ are relevant and we suppose that all other symmetries have been reduced out:
\begin{assumption}[No additional symmetries]\label{intro:assumption:no_additional_symmetries}
	\begin{enumerate}[(a)]
		\item Assume $W$ commutes with \emph{at most one} of the $T_j$, $j = 1 , 3$ and that there are no unitaries $U$ which commute with the material weights $W$ and the free Maxwell operator 
		\begin{align*}
			\Rot = \left (
			\begin{matrix}
				0 & + \ii \nabla^{\times} \\
				- \ii \nabla^{\times} & 0 \\
			\end{matrix}
			\right )
			.
		\end{align*}
		\item Assume $W$ commutes with $T_1$ \emph{and} $T_3$. Then we suppose that apart from a phase times $U_2 = \sigma_2 \otimes \id$ there are no \emph{other} unitaries $U$ which commute with the material weights $W$ and the free Maxwell operator $\Rot$. 
	\end{enumerate}
\end{assumption}
If there were an additional unitary, discrete, commuting symmetry $U$ present, then we first need to make a block decomposition of the Maxwell operator $M_+ = W^{-1} \, \Rot \, \big \vert_{\omega \geq 0}$, the analog of the quantum Hamiltonian defined by equation~\eqref{Schroedinger:eqn:positive_frequency_Maxwell_operator} below, with respect to the eigenspaces of $U$ and analyze each of the block operators separately. Depending on the interplay of all of $M_+$'s symmetries, the individual block operators may or may not inherit symmetries from $M_+$; 
we will give a detailed analysis of both of these cases in Section~\ref{conclusion:comparison:QSHE}. 
\medskip

\noindent
Therefore, we identify \emph{four} distinct types of media which correspond to the presence of no time-reversal symmetries ($1$), one even time-reversal symmetry ($2$) and two even time-reversal symmetries ($1$). 
\begin{theorem}[Symmetry classification of media]\label{intro:thm:classification_media}
	Suppose the material weights 
	\begin{align*}
		W(x) = \left (
		\begin{matrix}
			\eps(x) & \chi(x) \\
			\chi(x)^* & \mu(x) \\
		\end{matrix}
		\right )
		= \left (
		\begin{matrix}
			w_0(x) + w_3(x) & w_1(x) - \ii w_2(x) \\
			w_1(x) + \ii w_2(x) & w_0(x) - w_3(x) \\
		\end{matrix}
		\right )
		,
	\end{align*}
	expressed in terms of four hermitian $3 \times 3$ matrices $w_j(x) = w_j(x)^*$, $j = 0 , 1 , 2 , 3$, 
	are lossless and have strictly positive eigenvalues that are bounded away from $0$ and $\infty$, \ie they satisfy Assumption~\ref{Schroedinger:assumption:material_weights}. Moreover, we assume there are no additional unitary commuting symmetries (Assumption~\ref{intro:assumption:no_additional_symmetries}). Then there are four topologically distinct materials: 
	\begin{center}
		\newcolumntype{A}{>{\centering\arraybackslash\normalsize} m{2.85cm} }
		\newcolumntype{B}{>{\centering\arraybackslash\normalsize} m{2.3cm} }
		\newcolumntype{C}{>{\centering\arraybackslash\normalsize} m{2.475cm} }
		\newcolumntype{D}{>{\centering\arraybackslash\normalsize} m{1.7cm} }
		\newcolumntype{E}{>{\centering\arraybackslash\normalsize} m{1.45cm} }
		\renewcommand{\arraystretch}{1.15}
		\begin{tabular}{A | B | C | c | E}
			\emph{Material} & \emph{Realizations} & \emph{Conditions on $W$} & \emph{Symmetries} & \emph{CAZ Class} \\ \hline \hline 
			Dual symmetric materials \& vacuum & Yes \linebreak \cite{Fernandez-Corbaton_et_al:helicity_angular_momentum_dual_symmetry:2012,Fernandez-Corbaton_et_al:electromagnetic_duality_symmetry:2013} & $w_0 = \Re w_0$, $w_3 = 0$, $w_1 = 0$, $w_2 = \Re w_2$ & $T_1$, $T_3$, $U_2$ & not applicable \\ \hline
			Non-dual symmetric \linebreak \& non-gyrotropic & Yes \linebreak \cite{Bliokh_Bliokh:Berry_curvature_optical_Magnus_effect:2004,Onoda_Murakami_Nagaosa:Hall_effect_light:2004,Ochiai_Onoda:edge_states_photonic_graphene:2009} & $w_0 = \Re w_0$, $w_3 = \Re w_3$, $w_1 = \ii \, \Im w_1$, $w_2 = \Re w_2$ & $T_3$ & AI \\ \hline
			Magneto-electric & Yes \linebreak \cite{Tellegen:gyrator:1948,Lin_et_al:realization_magneto_electric_medium_static_fields:2008} & $w_0 = \Re w_0$, $w_3 = \ii \, \Im w_3$, $w_1 = \Re w_1$, $w_2 = \Re w_2$ & $T_1$ & AI \\ \hline
			Gyrotropic & Yes \linebreak \cite{Wang_et_al:edge_modes_photonic_crystal:2008,Lin_et_al:topological_photonic_states:2014} & 
			None & None & A \\ 
		\end{tabular}
	\end{center}
	The conditions on the material weights in each row are exclusive, meaning that \eg non-gyrotropic materials must violate at least one of the conditions that single out magneto-electric materials. 
\end{theorem}
Three of these four cases fall within the standard CAZ classification scheme and therefore closely resemble the corresponding quantum systems. Although, as we discuss below, drawing conclusions from that is not as easy as it might appear at first. Dual symmetric, non-gyrotropic media, the first case, falls outside of standard theory as it has two anti-commuting even time-reversal symmetries, and we will perform an analysis of the topological phases this topological class supports in Section~\ref{classification:Bloch_bundle:dual_symmetric}. 

At the end of the day, it turns out that of all periodic media in dimensions $1$, $2$ and $3$, \emph{only gyrotropic media in dimensions $2$ and $3$ can be topologically non-trivial}; the phases are labelled by one and three first Chern numbers, respectively. For periodic $3$d photonic crystals with periodic time-dependence ($d = 4$), second Chern classes will also play a role when distinguishing topologically distinct phases. More specifically, our classification result reads: 
\begin{theorem}[Topological bulk classification of periodic media]\label{intro:thm:bulk_classification}
	Suppose the material weights are periodic and satisfy Assumption~\ref{Schroedinger:assumption:material_weights}. 
	\begin{enumerate}[(1)]
		\item \textbf{Class A: Gyrotropic media} \\
		Phases are labelled by $\Z$-valued Chern numbers, in \\
		$d = 1$ by \emph{none} (topologically trivial), \\
		$d = 2$ by a \emph{single} first Chern number ($\Z$), \\
		$d = 3$ by \emph{three} first Chern numbers ($\Z^3$), \\
		$d = 4$ by \emph{six} first and \emph{one} second Chern number ($\Z^6 \oplus \Z$). 
		\item \textbf{Class AI: Non-dual symmetric, non-gyrotropic and magneto-electric media} \\
		In $d = 1 , 2 , 3$ these media are \emph{topologically trivial}, \ie there is a single phase. \\
		In $d = 4$, phases are labelled by a \emph{single} second Chern number ($\Z$). 
		\item \textbf{Dual-symmetric, non-gyrotropic media} \\
		In $d = 1 , 2 , 3$ these media are \emph{topologically trivial}, \ie there is a single phase. \\
		In $d = 4$, phases are labelled by \emph{two} second Chern numbers ($\Z^2$). 
	\end{enumerate}
\end{theorem}
\medskip

\noindent
\paragraph{Outline} 
\label{par:outline}
This work is separated into 5 Sections: following this introduction, we give an overview of the Schrödinger formalism for classical electromagnetism in media in Section~\ref{Schroedinger}, summarizing the results of \cite{DeNittis_Lein:Schroedinger_formalism_classical_waves:2017} to make this work more self-contained. After a discussion of the relevant material symmetries (Section~\ref{symmetries}), we give a bulk classification of electromagnetic media (Section~\ref{classification}). That includes a precise definition of what we mean by topology and how topological invariants are connected to relevant families of frequency bands. We close this work by contrasting and comparing it to the literature and outlining future developments (Section~\ref{conclusion}). 
\section{The Schrödinger formalism for electromagnetism in linear, dispersionless media} 
\label{Schroedinger}
The first step prior to adapting a quantum mechanical concept such as the symmetry classification is to write Maxwell's equations in Schrödinger form; this is a little more involved, and we have dedicated a separate paper \cite{DeNittis_Lein:Schroedinger_formalism_classical_waves:2017} to address all the intricacies that occur. These intricacies arise because we want to include \emph{gyrotropic} lossless media in our discussion where the material weights 
\begin{align}
	W(x) = \left (
	\begin{matrix}
		\eps(x) & \chi(x) \\
		\chi(x)^* & \mu(x) \\
	\end{matrix}
	\right )
	\label{Schroedinger:eqn:material_weights}
\end{align}
are complex, $W \neq \overline{W}$. Here, $W$ is a $6 \times 6$ matrix-valued function and is usually split into $3 \times 3$ blocks, the electric permittivity $\eps = \eps^*$, the magnetic permeability $\mu = \mu^*$ and the bianisotropic tensor $\chi$. It phenomenologically describes how the microscopic charges inside the material react to impinging electromagnetic fields. To be able to write Maxwell's equations~\eqref{Schroedinger:eqn:Maxwell_equations} we need to impose two conditions on the medium: 
\begin{assumption}[Material weights]\label{Schroedinger:assumption:material_weights}
	\begin{enumerate}[(a)]
		\item The medium is \emph{lossless}, \ie $W(x) = W(x)^*$ takes values in the \emph{hermitian} matrices. 
		\item The medium is \emph{not a negative index material}, \ie the eigenvalues 
		\begin{align*}
			0 < c \leq w_1(x) , \ldots , w_6(x) \leq C < \infty 
		\end{align*}
		of the hermitian matrix $W(x)$ are all positive and bounded away from $0$ and $\infty$ uniformly in $x$. 
	\end{enumerate}
\end{assumption}

\subsection{Maxwell's equations in linear, dispersionless media} 
\label{Schroedinger:Maxwells_equations}
In case the material weights $W \neq \overline{W}$ are complex, it is necessary split the physical field 
\begin{align*}
	\bigl ( \mathbf{E}(t) \, , \, \mathbf{H}(t) \bigr ) = \Psi_+(t) + \Psi_-(t) 
\end{align*}
into a complex wave 
\begin{align}
	\Psi_+(t) = \bigl ( \psi^E_+(t) \, , \, \psi^H_+(t) \bigr ) 
	= \frac{1}{\sqrt{2\pi}} \int_0^{\infty} \dd \omega \, \e^{- \ii t \omega} \, \bigl ( \widehat{\mathbf{E}}(\omega) \, , \, \widehat{\mathbf{H}}(\omega) \bigr ) 
	, 
	\label{Schroedinger:eqn:Psi_plus_dynamical_definition}
\end{align}
consisting only of \emph{positive} frequencies and the analogously defined \emph{negative} frequency contribution $\Psi_-$, defined in terms of the Fourier transformed fields $\bigl ( \widehat{\mathbf{E}}(\omega) \, , \, \widehat{\mathbf{H}}(\omega) \bigr )$. We will similarly have to split charge density $\pmb{\rho}(t) = \bigl ( \rho^D(t) \, , \, 0 \bigr ) = \rho_+(t) + \rho_-(t)$ and current density $\mathbf{J}(t) = \bigl ( \mathbf{J}^D(t) \, , \, 0 \bigr ) = J_+(t) + J_-(t)$ into positive and negative frequency parts. 

$\Psi_+$ and $\Psi_-$ evolve according to \emph{different} Maxwell equations, namely 
\begin{subequations}\label{Schroedinger:eqn:Maxwell_equations}
	\begin{empheq}[left={\omega \geq 0 \qquad \quad \empheqlbrace\,}]{align}
		\left (
		\begin{matrix}
			\eps & \chi \\
			\chi^* & \mu \\
		\end{matrix}
		\right )
		\, \frac{\partial}{\partial t} \left (
		\begin{matrix}
			\psi^E_+(t) \\
			\psi^H_+(t) \\
		\end{matrix}
		\right ) &= \left (
		\begin{matrix}
			+ \nabla \times \psi^H_+(t) \\
			- \nabla \times \psi^E_+(t) \\
		\end{matrix}
		\right ) - \left (
		\begin{matrix}
			J_+^D(t) \\
			0 \\
		\end{matrix}
		\right )
		, 
		\label{Schroedinger:eqn:Maxwell_equations:dynamics}
		\\
		\left (
		\begin{matrix}
			\nabla \cdot \bigl ( \eps \, \psi^E_+(t) + \chi \, \psi^H_+(t) \bigr ) \\
			\nabla \cdot \bigl ( \chi^* \, \psi^E_+(t) + \mu \, \psi^H_+(t) \bigr ) \\
		\end{matrix}
		\right ) &= \left (
		\begin{matrix}
			\rho_+^D(t) \\
			0 \\
		\end{matrix}
		\right )
		, 
		\label{Schroedinger:eqn:Maxwell_equations:constraint}
		\\
		\nabla \cdot J_+^D(t) + \partial_t \rho_+^D(t) &= 0 
		, 
		\label{Schroedinger:eqn:Maxwell_equations:charge_conservation}
	\end{empheq}
\end{subequations}
for non-negative frequencies and an analogous set of equations involving the complex conjugate weights $\overline{W}$ when $\omega \leq 0$. For otherwise $\bigl ( \mathbf{E}(t) \, , \, \mathbf{H}(t) \bigr )$ would acquire a non-vanishing imaginary part over time. 

To readers who would like to know why these equations, consisting of the \emph{dynamical equation}~\eqref{Schroedinger:eqn:Maxwell_equations:dynamics}, the \emph{constraint equation}~\eqref{Schroedinger:eqn:Maxwell_equations:constraint} and \emph{local charge conservation}~\eqref{Schroedinger:eqn:Maxwell_equations:charge_conservation}, are physically sensible, we refer to \cite[Section~2]{DeNittis_Lein:Schroedinger_formalism_classical_waves:2017} where we have derived them from Maxwell's equations for linear, \emph{dispersive} media. The main point is that the real-valuedness of the electromagnetic field as well as current and charge densities translates to 
\begin{align*}
	\bigl ( \widehat{\mathbf{E}}(\omega) \, , \, \widehat{\mathbf{H}}(\omega) \bigr ) = \bigl ( \overline{\widehat{\mathbf{E}}(-\omega)} \, , \, \overline{\widehat{\mathbf{H}}(-\omega)} \bigr )
\end{align*}
after Fourier transform and gives rise to \emph{phase locking condition} 
\begin{align}
	\Psi_- = \overline{\Psi_+} 
	\label{Schroedinger:eqn:phase_locking_condition}
\end{align}
for \emph{complex} waves. Put another way, positive and negative frequency contributions of the wave are \emph{not independent degrees of freedom}, if we know one, we can reconstruct the other. Hence, it suffices to consider \eqref{Schroedinger:eqn:Maxwell_equations} for $\omega \geq 0$ only. Implicitly, we have exploited that we can uniquely represent real fields with finite field energy as complex waves composed solely of non-negative frequencies (see \cite[Proposition~3.3]{DeNittis_Lein:Schroedinger_formalism_classical_waves:2017}); this systematic link to a \emph{complex} Hilbert space is essential if one wants to apply methods from quantum mechanics.%

\subsection{Rewriting Maxwell's equations in Schrödinger form} 
\label{Schroedinger:rewrite}
If we multiply both sides of \eqref{Schroedinger:eqn:Maxwell_equations:dynamics} by $\ii \, W^{-1}$ (which is bounded thanks to our assumptions on $W$), we obtain the \emph{Schrödinger form} of the dynamical law, 
\begin{align}
	\ii \, \partial_t \Psi_+(t) &= M_+ \Psi_+(t) - \ii \, W^{-1} \, J_+(t) 
	,
	&&
	\Psi_+(0) = Q_+ \, (\mathbf{E}_0 , \mathbf{H}_0)
	. 
	\label{Schroedinger:eqn:Maxwell_Schroedinger_equation}
\end{align}
Here, the complex positive frequency wave $\Psi_+$ represents the real electromagnetic field and plays the role of the wave function; the map $Q_+$ which connects $\Psi_+$ to the real fields $(\mathbf{E}_0 , \mathbf{H}_0)$ will be introduced below. The role of the Hamiltonian is played by the \emph{positive frequency Maxwell operator}
\begin{align}
	M_+ = W^{-1} \, \Rot \, \big \vert_{\omega \geq 0}
	\label{Schroedinger:eqn:positive_frequency_Maxwell_operator}
\end{align}
which is defined in terms of the \emph{free Maxwell operator}
\begin{align*}
	\Rot = \left (
	\begin{matrix}
		0 & + \ii \nabla^{\times} \\
		- \ii \nabla^{\times} & 0 \\
	\end{matrix}
	\right )
	. 
\end{align*}
$\nabla^{\times} \mathbf{E} = \nabla \times \mathbf{E}$ denotes the usual curl. 

Of course, we still need to specify the Hilbert space this operator acts on, and prove that $M_+$ is selfadjoint (or in physics parlance, hermitian). To do that, let us drop the frequency restrictions and define the \emph{auxiliary positive frequency Maxwell operator}
\begin{align*}
	\Maux_+ = W^{-1} \, \Rot
	. 
\end{align*}
Because $\Maux_+$ defines a selfadjoint (aka hermitian) operator on the Hilbert space 
\begin{align*}
	L^2_W(\R^3,C^6) = \Bigl \{ \Psi : \R^d \longrightarrow \C^n \; \; \big \vert \; \; \int_{\R^d} \dd x \, \Psi(x) \cdot W(x) \, \Psi(x) < \infty \Bigr \} 
\end{align*}
endowed with the \emph{energy scalar product}
\begin{align}
	\sscpro{\Phi}{\Psi}_W = \bscpro{\Phi}{W \, \Psi} 
	= \int_{\R^3} \dd x \, \Phi(x) \cdot W(x) \, \Psi(x) 
	, 
	\label{Schroedinger:eqn:weighted_scalar_product}
\end{align}
we can give meaning to the map 
\begin{align*}
	Q_+ = 1_{(0,\infty)}(\Maux_+) + \tfrac{1}{2} \, 1_{\{ 0 \}}(\Maux_+)
\end{align*}
with which we can uniquely represent real fields as complex waves composed solely of non-negative frequencies. The factor $\nicefrac{1}{2}$ is necessary so that gradient fields (which are static, \ie eigenfunctions to frequency $0$) are not counted twice (see \cite[Section~3.2.2]{DeNittis_Lein:Schroedinger_formalism_classical_waves:2017} for further explanations). 

Now $M_+ = \Maux_+ \, \big \vert_{\omega \geq 0}$ is the restriction of the auxiliary Maxwell operator to the non-negative frequency states, and $M_+$ acts on the Hilbert space 
\begin{align*}
	\Hil_+ = Q_+ \bigl [ L^2_W(\R^3,\C^6) \bigr ]
\end{align*}
of non-negative frequency states that inherits the energy scalar product; we need to include $\omega = 0$ waves, \ie gradient fields, so that we are able to cope with sources. What is more, there is a one-to-one correspondence between \emph{real} electromagnetic fields $(\mathbf{E},\mathbf{H}) \in L^2(\R^3,\R^6)$ with finite field energy and complex fields of non-negative frequencies (\cf \cite[Proposition~3.3]{DeNittis_Lein:Schroedinger_formalism_classical_waves:2017}), 
\begin{align*}
	(\mathbf{E},\mathbf{H}) \in L^2(\R^3,\R^6)
	\quad \longleftrightarrow \quad 
	\Psi_+ = Q_+ \, (\mathbf{E},\mathbf{H}) \in \Hil_+ 
	. 
\end{align*}
This systematic identification of real and complex fields allows us to employ technqiues from the theory of selfadjoint operators, and that necessarily forces us to work with \emph{complex} vector spaces. 

$M_+$ inherits the selfadjointness of $\Maux_+$, and therefore the evolution group $\e^{- \ii t M_+}$ is unitary with respect to $\scpro{\, \cdot \,}{\, \cdot \,}_W$. 
The unitarity implies the conservation of field energy in the absence of currents, 
\begin{align*}
	\mathcal{E} \bigl ( \mathbf{E}(t) , \mathbf{H}(t) \bigr ) &= \bscpro{\e^{- \ii t M_+} \Psi_+(0) \, }{ \, \e^{- \ii t M_+} \Psi_+(0)}_W
	= \bscpro{\Psi_+(0) \, }{ \, \Psi_+(0)}_W
	= \mathcal{E} \bigl ( \mathbf{E}(0) , \mathbf{H}(0) \bigr )
	, 
\end{align*}
thereby justifying the term \emph{energy scalar product} in the process. 

Lastly, there is the matter of the constraint equation~\eqref{Schroedinger:eqn:Maxwell_equations:constraint}. Provided local charge conservation~\eqref{Schroedinger:eqn:Maxwell_equations:charge_conservation} holds, the solution to \eqref{Schroedinger:eqn:Maxwell_Schroedinger_equation} automatically satisfies \eqref{Schroedinger:eqn:Maxwell_equations:constraint}. The key idea here is to decompose $\Psi_+$ into transversal and longitudinal components, and then use \eqref{Schroedinger:eqn:Maxwell_equations:charge_conservation}. 

\subsection{Symmetry between positive and negative frequency equations} 
\label{Schroedinger:symmetry}
For the symmetry arguments we also need to define the negative frequency counterparts, the Maxwell operator 
\begin{align*}
	M_- = \Maux_- \, \big \vert_{\omega \leq 0} 
	= \overline{W}^{\, -1} \, \Rot \, \big \vert_{\omega \leq 0}
	, 
\end{align*}
the projection onto non-positive states, $Q_- = 1_{(-\infty,0)}(\Maux_-) + \tfrac{1}{2} \, 1_{\{ 0 \}}(\Maux_-)$ and the Hilbert space $\Hil_-$ are defined with the complex conjugated weights $\overline{W}$. Complex conjugation $(C \Psi)(x) = \overline{\Psi(x)}$ defines an antiunitary between $\Hil_{\pm}$ and $\Hil_{\mp}$ that relates positive and negative frequency operators with one another, 
\begin{align*}
	M_- &= - C \, M_+ \, C 
	,
	\\
	Q_- &= C \, Q_+ \, C
	. 
\end{align*}
However, \emph{complex conjugation can never be a physically meaningful symmetry of the Maxwell operator}, because the real-valuedness of the physical fields is an unbreakable tenet of classical electromagnetism and the resulting phase locking condition~\eqref{Schroedinger:eqn:phase_locking_condition} implies that positive and negative frequency fields are not independent degrees of freedom. 

\subsection{Maxwell's equations in other dimensions} 
\label{Schroedinger:other_dimensions}
Maxwell's equations are naturally defined in three spatial dimensions, but electromagnetic waves can be confined to lower-dimensional media. One common way to obtain those lower-dimensional media is to consider \emph{wave guides} where in one or two directions the medium is terminated by a metal. Conversely, there are cases when the effective dimension of the system may exceed $3$. 

We emphasize that \emph{all} the symmetries that we will consider in the next section do not depend on $x$, derivatives $\nabla$ or time $t$, and only impose conditions such as $\eps = \Re \eps$ or $\eps = \mu$. These symmetries then define (anti)unitaries on the relevant Hilbert space on which the (auxiliary) Maxwell operator is defined. Again, we postpone the technical details to a future work.

\subsubsection{Lower-dimensional electromagnetic media} 
\label{Schroedinger:other_dimensions:lower}
The two-dimensional gyrotropic photonic crystal realized in \cite{Wang_et_al:edge_modes_photonic_crystal:2008} which exhibited topologically protected edge modes is an example of a lower-dimensional electromagnetic medium. Here, YIG rods (immersed in a constant magnetic field to tune the material weights) were arranged in a quadratic lattice and sandwiched between two metal plates, thereby forming a waveguide for microwaves with periodic interior. The height $h \approx \unit[7]{mm}$ of the wave guide (which we take to point in the $z$-direction) was comparable to the lattice length $l \approx \unit[40]{mm}$. Usually, the metal walls are approximated by an idealized perfect electric conductor (PEC) where appropriate boundary conditions 
\begin{subequations}\label{Schroedinger:eqn:PEC_boundary_conditions}
	\begin{align}
		\mathbf{E}_x(x,0) &= 0 = \mathbf{E}_x(x,h) 
		, 
		\\
		H_z(x,0) &= 0 = H_z(x,h) 
		, 
	\end{align}
\end{subequations}
are imposed on the electromagnetic field. As the notation suggests, $\mathbf{E}_x = (E_{x_1} , E_{x_2})$ is the in-plane component of the electric field $\mathbf{E} = (\mathbf{E}_x , E_z)$ and $H_z$ the $z$-component of the magnetic field $\mathbf{H} = (\mathbf{H}_x , H_z)$. We will forgo a precise mathematical definition (which is straight-forward, but technical), and only sketch the strategy of defining first the auxiliary, then the physical Maxwell operator. The weights $W(x,z)$ evidently only need to be periodic in the $x$-direction with respect to a two-dimensional periodicity lattice. Straight-forward arguments show that $\Maux_+ = W^{-1} \, \Rot$, endowed with the proper domain, is a selfadjoint (hermitian) operator, and we may impose the condition $\omega \geq 0$ just like before via the projection 
\begin{align*}
	Q_+ = 1_{(0,\infty)}(\Maux_+) + \tfrac{1}{2} \, 1_{\{ 0 \}}(\Maux_+)
\end{align*}
onto the physical states with frequencies $\omega \geq 0$. This gives rise to $M_+ = \Maux_+ \, \vert_{\omega \geq 0}$ and the Hilbert space 
\begin{align*}
	\Hil_+ = \ran Q_+ 
	\subset L^2 \bigl ( \R^2 \times [0,h] , \C^6 \bigr ) 
	. 
\end{align*}
The situation further simplifies if we assume that 
\begin{align*}
	W(x,z) = \left (
	\begin{matrix}
		\eps_x(x) & 0 & 0 & 0 \\
		0 & \eps_z(x) & 0 & 0 \\
		0 & 0 & \mu_x(x) & 0 \\
		0 & 0 & 0 & \mu_z(x) \\
	\end{matrix}
	\right )
\end{align*}
is independent of $z$ and $\eps$ and $\mu$ split cleanly into $x$ and $z$ components. Then we can express every electromagnetic wave as a linear combination of plane waves $\bigl ( \mathbf{E}(x,k_z) , \mathbf{H}(x,k_z) \bigr ) \, \e^{+ \ii k_z \, z}$ in $k_z$, where of course, $k_z$ may only take discrete values. Put another way, the discrete Fourier transform $\Fourier_z$ in the $z$-direction decomposes the auxiliary Maxwell operator 
\begin{align*}
	\Fourier_z \, M_+ \, \Fourier_z^{-1} = \bigoplus_{k_z \in \frac{2\pi}{h} \Z} M_+(k_z)
\end{align*}
into a direct sum of operators $M_+(k_z)$, each associated to a fixed momentum $k_z$ that governs the dynamics of $\bigl ( \mathbf{E}(x,k_z) , \mathbf{H}(x,k_z) \bigr )$. In the aforementioned experiment \cite{Wang_et_al:edge_modes_photonic_crystal:2008} only the $k_z = 0$ mode contributed, and it suffices to consider 
\begin{align*}
	M_+(0) &= \left . W(x) \, \left (
	\begin{matrix}
		0 & + \ii (\nabla_x , 0)^{\times} \\
		- \ii (\nabla_x , 0)^{\times} & 0 \\
	\end{matrix}
	\right ) \right \vert_{\omega \geq 0} 
\end{align*}
acting the positive frequency subspace of $L^2_W(\R^2,\C^6)$; this operator has a more compact expression as a $3 \times 3$-matrix-valued operator (see \eg \cite[Section~2.4]{DeNittis_Lein:symmetries_Maxwell:2014}). We can adapt all of our arguments without any essential changes, \eg employ the Bloch-Floquet transform in $x = (x_1 , x_2)$ in order to obtain frequency bands and Bloch functions that now depend on $k_x = (k_{x_1} , k_{x_2})$ (\cf Section~\ref{classification:Bloch_Floquet}). 

The construction of the Maxwell operator for quasi-one-dimensional waveguides is analogous. Evidently, the more complicated the waveguide geometry is, the less explicit expressions we get. 

\subsubsection{Time-dependent media and media with synthetic dimensions} 
\label{Schroedinger:other_dimensions:higher}
Photonic crystals which are modulated periodically in time could also be treated within this framework by making use of techniques developed for time-dependent quantum systems (see \eg \cite[Section~4.4]{PST:sapt:2002} or \cite{King-Smith_Vanderbilt:polarization:1993,Resta:electric_polarization:1992,Panati_Sparber_Teufel:polarization:2006,DeNittis_Lein:piezo_graphene:2013}). Time then appears alongside $k = (k_1 , \ldots , k_d)$ as a periodic variable, and for topological considerations the system becomes $d+1$-dimensional. 

Physical systems can also be designed to have \emph{synthetic dimensions} by making the system parameter-dependent or adding internal degrees of freedom (see \eg \cite{Ozawa_et_al:synthetic_dimensions_photonics_4d_quantum_Hall_effect:2016,Zilberberg_et_al:4d_quantum_Hall_effect_waveguide_array:2017}).%
\section{Discrete material symmetries of electromagnetic media} 
\label{symmetries}
The standard classification scheme for topological insulators \cite{Altland_Zirnbauer:superconductors_symmetries:1997,Schnyder_Ryu_Furusaki_Ludwig:classification_topological_insulators:2008}, also known as the Ten-Fold Way or the Cartan-Altland-Zirnbauer (CAZ) scheme, distinguishes 10 different topological classes. Which topological class a system belongs to is determined by the symmetries of the Hamilton or Maxwell operator $M$ which enters the dynamical equation $\ii \partial_t \Psi = M \Psi$. Inside of each topological class there are inequivalent phases labeled by a finite set of topological invariants such as Chern numbers or the Kane-Mele invariant. The number and nature of these invariants depends crucially on the symmetries and the dimensionality of the system. 

We have applied this scheme in a previous work \cite{DeNittis_Lein:symmetries_Maxwell:2014} to Maxwell's equations, but as mentioned in the introduction, the equations we used for media with complex material weights included states that were unphysical. Now that we have rewritten these equations in Schrödinger form, we have finished all preparations to have a physically meaningful classification of PTIs. 

Compared to the classification theory developed for quantum mechanics, there is one major difference: \emph{electromagnetic fields are real.} Hence, it is not clear whether we are able to employ the standard classification machinery developed for \emph{complex} vector spaces. The systematic identification of real fields with complex wave functions that was part and parcel of the Schrödinger formalism allows us to overcome this conceptual chasm. We emphasize that all of what we do in this section \emph{applies to homogeneous, random and periodic media} alike, covers lower-dimensional waveguides and time-dependent media\footnote{To simplify our presentation, we will not make the time-dependence explicit in case the medium changes periodically in time. }, and is a prerequisite to the symmetry classification of periodic light conductors in Section~\ref{classification}.

\subsection{Relevant symmetries} 
\label{symmetries:relevant}
For the purpose of classifying photonic topological insulators, we are interested in four basic types of symmetries: \emph{unitary} operators with $U^2 = + \id$ are called \emph{regular} symmetries if 
\begin{align}
	U \, M \, U^{-1} = + M
	, 
	\label{symmetries:eqn:commuting_symmetry}
\end{align}
and \emph{chiral} (pseudo) symmetries in case
\begin{align}
	U \, M \, U^{-1} = - M
	. 
	\label{symmetries:eqn:anticommuting_symmetry}
\end{align}
\emph{Anti}unitaries $U$ are said to be of \emph{time-reversal}-type if \eqref{symmetries:eqn:commuting_symmetry} holds, and are \emph{particle-hole}-type symmetries if just as in \eqref{symmetries:eqn:anticommuting_symmetry} $U$ \emph{anti}commutes with $M$. Regardless of whether they commute or anticommute, antiunitaries come in the even and the odd variety, depending on whether $U^2 = \pm \id$. In what follows, we will refer to all four merely as \emph{symmetries} unless the distinction between \emph{proper}, \ie commuting, symmetries and anticommuting \emph{pseudo} symmetries becomes important. 

Let us emphasize that the terminology originates from quantum mechanics and should \emph{not} be taken literally in this context. The presence of a particle-hole symmetry, \ie an anticommuting antiunitarity, does \emph{not} necessarily postulate the existence of particles and antiparticles. 

The notable absence of commuting unitaries in the CAZ scheme rests on the \emph{assumption} that all commuting unitary symmetries commute with all other types of symmetries listed above. This allows us to reduce out all of these commuting unitary symmetries first and consider the block decomposition with respect to these commuting symmetries; the block operators retain all of the other (time-reversal, particle-hole and chiral) symmetries. However, \emph{a priori} we cannot be sure whether the assumption that regular unitary symmetries commute with the others, this is something that remains to be checked on a case-by-case basis. 

The form of the relevant symmetries is suggested by the problem: we can express the weights 
\begin{align*}
	W = \left (
	\begin{matrix}
		\eps & \chi \\
		\chi^* & \mu \\
	\end{matrix}
	\right )
	= \sum_{j = 0}^3 \sigma_j \otimes w_j
\end{align*}
and the free Maxwell operator 
\begin{align*}
	\Rot = - \sigma_2 \otimes \nabla^{\times} 
\end{align*}
in terms of the identity $\sigma_0 = \id$, the Pauli matrices $\sigma_1$, $\sigma_2$ and $\sigma_3$, and $3 \times 3$ block operators acting on the electric or magnetic fields. The first one $\nabla^{\times} \psi^E = \nabla \times \psi^E$ is just the usual curl. Electric permittivity $\eps$ and magnetic permeability $\mu$ determine $w_0 = \tfrac{1}{2} (\eps + \mu)$ and $w_3 = \tfrac{1}{2} (\eps - \mu)$; hermitian and antihermitian parts of the bianisotropic tensor $\chi$ fix $w_1 = \tfrac{1}{2} (\chi + \chi^*)$ and $w_2 = \tfrac{\ii}{2} (\chi - \chi^*)$. 

Therefore, we shall consider either 
\begin{align}
	U_n = \sigma_n \otimes \id
	, 
	\qquad \qquad 
	n = 1 , 2 , 3
	, 
	\label{symmetries:eqn:linear_symmetry}
\end{align}
as candidates for \emph{linear} symmetries and 
\begin{align}
	T_n = (\sigma_n \otimes \id) \, C 
	, 
	\qquad \qquad 
	n = 0 , 1 , 2 , 3
	, 
	\label{symmetries:eqn:antilinear_symmetry}
\end{align}
for \emph{anti}linear symmetries. $T_1 = \left (
\begin{smallmatrix}
	0 & \id \\
	\id & 0 \\
\end{smallmatrix}
\right ) \, C$ exchanges electric and magnetic fields, then complex conjugates them, and captures whether the medium treats electric and magnetic fields differently; this symmetry will play a role in our analysis later on. Symmetries in our case have to satisfy two conditions which \emph{seemingly} have to be imposed in addition to $U \, M_+ \, U^{-1} = \pm M_+$, namely (1)~$U$ maps non-negative frequency states onto non-negative frequency states and (2)~it needs to satisfy $U^{\ast_W} = W^{-1} \, U^* \, W = U^{-1}$, where $U^{\ast_W}$ denotes the adjoint with respect to the \emph{weighted} scalar product~\eqref{Schroedinger:eqn:weighted_scalar_product} (\cf \cite[Section~3.1.2]{DeNittis_Lein:Schroedinger_formalism_classical_waves:2017}); this adjoint is to be distinguished from $U^*$, the adjoint with respect to the usual, unweighted scalar product. It turns out \emph{only commuting unitaries or time-reversal-type symmetries are admissible.} 
\begin{lemma}[Conditions for admissibility of symmetries]\label{symmetries:lem:symmetry_condition}
	Suppose the material weights $W$ satisfy Assumption~\ref{Schroedinger:assumption:material_weights}. Then all admissible symmetries $U : \Hil_+ \longrightarrow \Hil_+$ of the form \eqref{symmetries:eqn:linear_symmetry} or \eqref{symmetries:eqn:antilinear_symmetry} must commute with 
	\begin{enumerate}[(1)]
		\item the material weights $W$ (due to $\scpro{\; \cdot \,}{\, \cdot \;}_W$-\emph{unitarity}), and 
		\item the Maxwell operator $M_+$ or, equivalently, with $\Maux_+$ or $\Rot$ (due to \emph{non-negative states are mapped onto non-negative states}). 
	\end{enumerate}
	Analogous statements hold for the case $\omega \leq 0$ with weights $\overline{W}$ and $M_- = \overline{W} \, \Rot \, \big \vert_{\omega \leq 0}$. 
\end{lemma}
\begin{proof}
	For simplicity, let us only consider the case of non-negative frequencies and write (anti)unitary to mean either unitary or antiunitary. 
	
	Let us initially work on $L^2_W(\R^3,\C^6)$, the Hilbert space prior to restricting to $\omega \geq 0$. Because all of the operators $U$ are defined in terms of the Pauli matrices, they not only preserve the domain of $\Maux_+$, $U \, \mathcal{D}(\Maux_+) = \mathcal{D}(\Maux_+)$, they are also automatically (anti)unitaries with respect to the ordinary scalar product. Therefore, the ordinary adjoint 
	\begin{align*}
		U^* = U^{-1} \overset{!}{=} U^{\ast_W} = W^{-1} \, U^* \, W 
	\end{align*}
	necessarily agrees with the weighted adjoint $U^{\ast_W} = W^{-1} \, U^* \, W$ on $L^2_W(\R^3,\C^6)$, and this means that $U$ needs to commute with $W$ (condition~(1)). 
	
	Initially, these equalities hold on $L^2_W(\R^3,\C^6)$, and we have to check whether $U : \Hil_+ \longrightarrow \Hil_+$, restricted to the non-negative frequencies, is well-defined and still (anti)unitary. Mathematically, we have to impose $[U , Q_+] = 0$ or, alternatively, 
	\begin{align*}
		U \, Q_+ \, U^{-1} &= U \, \Bigl ( 1_{(0,\infty)}(\Maux_+) + \tfrac{1}{2} \, 1_{\{0\}}(\Maux_+) \Bigr ) \; U^{-1} 
		\\
		&= 1_{(0,\infty)} \bigl ( U \, \Maux_+ \, U^{-1} \bigr ) + \tfrac{1}{2} \, 1_{\{0\}} \bigl ( U \, \Maux_+ \, U^{-1} \bigr )
		\\
		&\overset{!}{=} Q_+ = 1_{(0,\infty)}(\Maux_+) + \tfrac{1}{2} \, 1_{\{0\}}(\Maux_+) 
		. 
	\end{align*}
	Thus, not only the domain of $\Maux_+$ but also that of $M_+$ is preserved, 
	\begin{align*}
		U \, \mathcal{D}(M_+) = U \, Q_+ \, \mathcal{D}(\Maux_+)
		= Q_+ \, U \, \mathcal{D}(\Maux_+) 
		= \mathcal{D}(M_+)
	\end{align*}
	and we conclude that anticommuting symmetries are forbidden, only commuting symmetries are admissible. Furthermore, due to the product structure of $\Maux_+ = W \, \Rot$ and $[U , W] = 0$, this is equivalent to saying $[U , \Rot] = 0$, and we have shown condition~(2). 
	
	Evidently, these arguments hold verbatim for $M_- = \overline{W} \, \Rot \, \big \vert_{\omega \leq 0}$ after replacing $W$ with $\overline{W}$ and $Q_+$ with $Q_-$. What is more, because complex conjugation relates $M_+$ and $M_- = - C \, M_+ \, C$ (including $C : \Hil_{\pm} \longrightarrow \Hil_{\mp}$), if $U$ is a symmetry of $M_+$, then $C \, U \, C$ is a symmetry of $M_-$. 
\end{proof}
%

\subsection{Conditions on $\eps$, $\mu$ and $\chi$} 
\label{symmetries:translation}
Now we check which of the symmetries~\eqref{symmetries:eqn:linear_symmetry} and \eqref{symmetries:eqn:antilinear_symmetry} is admissible and translate their presence to conditions on $\eps$, $\mu$ and $\chi$. With the help of Lemma~\ref{symmetries:lem:symmetry_condition} and the algebra of Pauli matrices, we can summarize the results as follows: 
\begin{proposition}[Symmetry conditions on the $w_j$]\label{symmetries:prop:admissible_symmetries}
	Suppose $W$ satisfies Assumption~\ref{Schroedinger:assumption:material_weights}. Of the 7 operators considered in \eqref{symmetries:eqn:linear_symmetry} and \eqref{symmetries:eqn:antilinear_symmetry} only three are admissible, and their presence ($[ T_j , W ] = 0$) translates to the following conditions on the material weights: 
	\begin{center}
		\renewcommand{\arraystretch}{1.25}
		\begin{tabular}{c | c | c | c | c | c}
			Symmetry & $w_0 = $ & $w_1 = $ & $w_2 = $ & $w_3 = $ & \emph{Symmetry Type} \\ \hline\hline
			$T_1 = (\sigma_1 \otimes \id) \, C$ & $\Re w_0$   & $\Re w_1$   & $\Re w_2$ & $\ii \, \Im w_3$ & +TR \\ \hline
			$U_2 = \sigma_2 \otimes \id$        & $w_0$       & $0$         & $w_2$     & $0$           & unitary, commuting \\ \hline
			$T_3 = (\sigma_3 \otimes \id) \, C$ & $\Re w_0$   & $\ii \, \Im w_1$   & $\Re w_2$ & $\Re w_3$ & +TR \\
		\end{tabular}
	\end{center}
	All of them are either proper symmetries or even time-reversal symmetries. 
\end{proposition}
Note that since $T_1 \, T_3 = \ii U_2$ the presence of any two symmetries implies the presence of the third. 
\begin{proof}
	We have to check conditions enumerated in Lemma~\ref{symmetries:lem:symmetry_condition}: Computing the signs of $U \, \Rot \, U^{-1} = \pm \Rot$ is straightforward. Imposing $U \, W \, U^{-1} = W$ then singles out those three above, because only they commute with $\Rot$ and hence, lead to $U \, M_+ \, U^{-1} = + M_+$. The conditions on the $w_j$ listed in the table can then be obtained by comparing $U \, W \, U^{-1}$ and $W$. 
\end{proof}
%

\subsection{Four topologically distinct media} 
\label{symmetries:experimental_realizations}
At least three experimentally realized materials fall within our classification scheme; though the first two can be understood within the ordinary CAZ classification scheme, the last belongs to none of the 10 classes, and could therefore exhibit novel topological effects. 
\begin{theorem}[Symmetry classification of media]\label{symmetries:thm:classification_media}
	Suppose the material weights satisfy Assumption~\ref{Schroedinger:assumption:material_weights}. Then there are four topologically distinct electromagnetic media: 
	\begin{center}
		\renewcommand{\arraystretch}{1.25}%
		\newcolumntype{A}{>{\centering\arraybackslash\normalsize} m{30mm} }
		\begin{tabular}{A | c | c | c}
			\emph{Material} & \emph{Realizations} & \emph{Symmetries} & \emph{CAZ Class} \\ \hline \hline 
			Dual symmetric, non-gyrotropic & vacuum, \cite{Fernandez-Corbaton_et_al:helicity_angular_momentum_dual_symmetry:2012,Fernandez-Corbaton_et_al:electromagnetic_duality_symmetry:2013} & $T_1$, $T_3$, $U_2$ & N/A \\ \hline 
			Non-dual symmetric, non-gyrotropic & \cite{Bliokh_Bliokh:Berry_curvature_optical_Magnus_effect:2004,Onoda_Murakami_Nagaosa:Hall_effect_light:2004,Ochiai_Onoda:edge_states_photonic_graphene:2009} & $T_3$ & AI \\ \hline
			Magneto-electric & \cite{Tellegen:gyrator:1948,Lin_et_al:realization_magneto_electric_medium_static_fields:2008} & $T_1$ & AI \\ \hline
			Gyrotropic & \cite{Wang_et_al:edge_modes_photonic_crystal:2008,Lin_et_al:topological_photonic_states:2014} & None & A \\ 
		\end{tabular}
	\end{center}
	The specific conditions which arise from imposing $[T_j , W] = 0$ can be read off of the table in Proposition~\ref{symmetries:prop:admissible_symmetries}. 
\end{theorem}
One of our main motivations is to give a first-principles explanation of the Quantum Hall Effect for Light \cite{Raghu_Haldane:quantum_Hall_effect_photonic_crystals:2008,Wang_et_al:unidirectional_backscattering_photonic_crystal:2009}. The above result sheds some light on its inner workings: while making $W$ complex is the right thing to do, the symmetry to be broken is the even time-reversal symmetry $T_3 = (\sigma_3 \otimes \id) \, C$ rather than complex conjugation $C$ as claimed in \cite{Raghu_Haldane:quantum_Hall_effect_photonic_crystals:2008}. In fact, $C$ is \emph{not even an admissible symmetry} in the Schrödinger formalism, since it swaps positive and negative frequency states. Even in the non-gyrotropic case where $W = \overline{W}$, complex conjugation acts as an even \emph{particle-hole}-type symmetry of the auxiliary Maxwell operator $\Maux_+ = - C \, \Maux_+ \, C$, and therefore no matter the framework the nature of $C$ can never be of time-reversal type. Evidently, breaking time-reversal $T_3$ is a \emph{necessary} condition for the existence of \emph{unidirectional} edge modes. For otherwise modes come in counterpropagating pairs related by time-reversal. Before applying these results to photonic crystals, though, we will briefly consider the relationship of symmetries and sources. 

\subsection{Symmetries imposed on sources} 
\label{symmetries:sources}
Apart from the reality of the physical fields, a second difference to quantum mechanics is the potential presence of sources in the Schrödinger equation~\eqref{Schroedinger:eqn:Maxwell_Schroedinger_equation}, and we may ask what role symmetries play here. From our discussion above, we only need to consider ordinary symmetries and those of time-reversal type. 

The idea here is to generalize the condition 
\begin{align*}
	U \, \e^{- \ii t M_+} = \e^{- \ii (\pm t) M_+} \, U
	, 
\end{align*}
with $+$ chosen for linear, commuting symmetries and $-$ for time-reversal symmetries, to solutions of the equation with sources, 
\begin{align*}
	(\mathcal{U} \Phi)(t) &= \e^{- \ii t M_+} \Phi - \ii \int_0^t \dd s \, \e^{- \ii (t-s) M_+} \, W^{-1} \, J_+(s) 
	. 
\end{align*}
Imposing $U (\mathcal{U} \Phi)(t) = \bigl ( \mathcal{U}(U \Phi) \bigr )(\pm t)$ yields $U \, J_+(t) = J_+(\pm t)$, \ie the signs need to match. 
\section{Bulk classification of topological photonic crystals} 
\label{classification}
Generally, we need \emph{two} ingredients to create topological effects: (1)~We need to break or preserve the right \emph{symmetries}. And (2), we need a \emph{spectral gap}. In electromagnetic media, gaps can either be created by periodic patterning \cite{Yablonovitch:photonic_band_gap:1993} or by using dispersion \cite{Silveirinha:Z2_topological_index_continuous_photonic_systems:2016}. While Silveirinha's recent works \cite{Silveirinha:Z2_topological_index_continuous_photonic_systems:2016,Gangaraj_Silveirinha_Hanson:Berry_connection_Maxwell:2017} make first steps to classify homogeneous, dispersive media, we focus on photonic crystals, \ie electromagnetic media with periodic structure. 

Phase relations are at the root of all topological effects, and one way to encode them is to construct vector bundles. If necessary, these vector bundles are endowed with symmetries that are inherited from the Maxwell operator. This follows the exact same playbook as in the quantum case, pioneered by \cite{Thouless_Kohmoto_Nightingale_Den_Nijs:quantized_hall_conductance:1982}. Just like in the theory of crystalline solids, the Bloch bundle is obtained from a collection of Bloch waves which arise naturally in the context of periodic systems.

\subsection{The frequency band picture} 
\label{classification:Bloch_Floquet}
Let us start with a time-independent three-dimensional medium where $d = 3$, the dimension where Maxwell's are naturally defined. The periodicity of the weights $W$ and of the Maxwell operator $M_+$ with respect to the lattice $\Gamma \cong \Z^3$ can be exploited via the Bloch-Floquet-Zak transform 
\begin{align}
	(\Fourier \Psi)(k,x) = \sum_{\gamma \in \Gamma} \e^{- \ii k \cdot (x + \gamma)} \, \Psi(x + \gamma)
	\label{classification:eqn:Bloch_Floquet_transform}
\end{align}
which maps onto the space-periodic part of the Bloch functions. This is a standard tool in the theory of periodic operators and has been applied to great effect to various types of equations \cite{Grosso_Parravicini:solid_state_physics:2003,Kuchment:Floquet_theory:1993,Kuchment:math_photonic_crystals:2001,DeNittis_Lein:adiabatic_periodic_Maxwell_PsiDO:2013}. $\Fourier$ defines a unitary map between the Hilbert spaces before restricting to non-negative frequencies, $L^2_W(\R^3,\C^6)$ and $L^2(\BZ) \otimes L^2_W(\T^3,\C^6)$; the first factor $L^2(\BZ)$ is the usual, unweighted $L^2$-space over the Brillouin torus $\BZ$ and $L^2_W(\T^3,\C^6)$ is the Hilbert space over the Wigner-Seitz cell, also seen as a torus, and endowed with a weighted scalar product akin to \eqref{Schroedinger:eqn:weighted_scalar_product}, 
\begin{align}
	\bscpro{\phi(k)}{\psi(k)}_W &= \bscpro{\phi(k)}{W \, \psi(k)} 
	= \int_{\T^3} \dd x \, \phi(k,x) \cdot W(x) \psi(k,x) 
	. 
	\label{classification:eqn:weighted_scalar_product_torus}
\end{align}
Both, the auxiliary and the frequency constrained Maxwell operators are periodic and therefore admit a fiber decomposition in Bloch momentum $k \in \BZ$. Starting with the \emph{auxiliary} Maxwell operator, we see that 
\begin{align*}
	\Fourier \, \Maux_+ \, \Fourier^{\, -1} = \int_{\BZ}^{\oplus} \dd k \, \Maux_+(k) 
\end{align*}
consists of a \emph{collection of operators}
\begin{align*}
	\Maux_+(k) = W \, \Rot(k) 
	= \left (
	\begin{matrix}
		\eps & \chi \\
		\chi^* & \mu \\
	\end{matrix}
	\right ) \, \left (
	\begin{matrix}
		0 & - (- \ii \nabla + k)^{\times} \\
		+ (- \ii \nabla + k)^{\times} & 0 \\
	\end{matrix}
	\right )
\end{align*}
acting on $L^2_W(\T^3,\C^6)$, the Hilbert space associated to electromagnetic fields defined on the unit cell in real space. Here, the operator $v^{\times} \mathbf{E} = v \times \mathbf{E}$ denotes the matrix form associated to the crossed product with any vectorial quantity $v = (v_1 , v_2 , v_3)$ from the left. Following the procedure of the quantum case, we arrive at the \emph{frequency band picture} by looking at solutions to 
\begin{align*}
	\Maux_+(k) \varphi_n(k) = \omega_n(k) \, \varphi_n(k)
\end{align*}
where $\varphi_n(k) = \bigl ( \varphi_n^E(k) , \varphi_n^H(k) \bigr )$ is a (necessarily complex) Bloch function and $\omega_n(k)$ an eigenvalue. 

Properties of $\Maux_+(k)$ have been investigated extensively in the past (\eg in \cite{Kuchment:math_photonic_crystals:2001,DeNittis_Lein:adiabatic_periodic_Maxwell_PsiDO:2013}), and there are a few features of note that set it apart from the condensed matter case: first of all, the longitudinal gradient fields contribute an infinitely degenerate flat band $\omega_0(k) = 0$, and these bands only play a role if sources are present. In the absence of charge densities, electromagnetic waves that satisfy~\eqref{Schroedinger:eqn:Maxwell_equations:constraint} are necessarily transversal. Apart from the infinitely degenerate eigenvalue $0$, all other (positive \emph{and} negative!) eigenvalues $\omega_n(k)$ that make up the frequency bands have finite degeneracy, and their Bloch functions span the space $L^2_W(\T^3,\C^6)$. The analyticity of $\Maux_+(k)$ in $k$ (the operator is linear in $k$ and its domain is independent of $k$ \cite[p.~68]{DeNittis_Lein:adiabatic_periodic_Maxwell_PsiDO:2013}) means that these eigenvalues form (frequency) bands. By convention, the flat band $\omega_0(k) = 0$ due to the gradient fields has band index $0$, frequency bands for which $\omega_n(k) > 0$ when $k \neq 0$ are enumerated with positive integers $n > 0$ while negative indices $n < 0$ are reserved for the unphysical Bloch functions with $\omega_n(k) < 0$ for $k \neq 0$. 

Secondly, there are always \emph{ground state bands}, \ie two positive and two negative frequency bands (including degeneracy) with approximately linear dispersion at $k \approx 0$ and $\omega \approx 0$. (That is why we had to exclude $k = 0$ when labeling frequency bands.) The ground state Bloch functions are necessarily discontinuous at $k = 0$, since the transversality condition degenerates there and consequently, the dimensionality if the eigenspace changes by $2$ \cite[Section~3.2–3.3]{DeNittis_Lein:adiabatic_periodic_Maxwell_PsiDO:2013}. Apart from the ground state bands, though, both Bloch functions (after a judicious choice of phase) and frequency band functions are locally analytic away from band crossings; in these respects, they mimic the Schrödinger case. 

The physical Maxwell operator is the restriction of $\Maux_+$ to complex waves with $\omega \geq 0$, \ie we discard unphysical waves of frequency $\omega < 0$; these waves are unphysical, because the physically relevant negative frequency waves are subject to the complex conjugate weights $\overline{W}$ \cite[Section~2.3]{DeNittis_Lein:Schroedinger_formalism_classical_waves:2017}. For periodic systems, we can make this restriction more explicit: The Bloch waves associated to the physically relevant Bloch bands $\omega_n(k) \geq 0$ span the Hilbert space of physically relevant waves 
\begin{align*}
	\Hil_+(k) &= \Jphys(k) \oplus \mathcal{G}(k) 
	\subset L^2_W(\T^3,\C^6)
\end{align*}
where we distinguish between the (positive frequency) transversal fields 
\begin{align*}
	\Jphys(k) &= \mathrm{span} \Bigl \{ \varphi \in L^2_W(\T^3,\C^6) \; \; \big \vert \; \; \Maux_+(k) \varphi = \omega_n(k) \, \varphi , \; \; n > 0 \Bigr \} 
\end{align*}
and the longitudinal gradient fields 
\begin{align*}
	\mathcal{G}(k) = \Bigl \{ \varphi \in L^2_W(\T^3,\C^6) \; \; \big \vert \; \; \Maux_+(k) \varphi = 0 \Bigr \} 
	. 
\end{align*}
The symbol $\oplus$ for the orthogonal sum means that we can uniquely write any $\Psi = \Psi_{\perp} + \Psi_{\parallel} \in \Hil_+(k)$ as the sum of a transversal field $\Psi_{\perp}(k) \in \Jphys(k)$ and a longitudinal field $\Psi_{\parallel}(k) \in \mathcal{G}(k)$ that are orthogonal to each other with respect to the weighted scalar product~\eqref{classification:eqn:weighted_scalar_product_torus}. The restriction of the \emph{auxiliary} operator $\Maux_+(k)$ to the physically relevant subspace of fields with non-negative frequencies yields the $k$-dependent \emph{Maxwell operator}
\begin{align*}
	M_+(k) = \Maux_+(k) \big \vert_{\Hil_+(k)}
	, 
\end{align*}
whose Bloch functions are those of $\Maux_+(k)$ for non-negative frequency bands. 

Symbolically we write that the Bloch-Floquet-Zak transform is a unitary between the Hilbert spaces 
\begin{align*}
	\Fourier : \Hil_+ \longrightarrow \int_{\BZ}^{\oplus} \dd k \, \Hil_+(k) 
\end{align*}
that fiber-decomposes the Maxwell operator 
\begin{align*}
	\Fourier \, M_+ \, \Fourier^{-1} = \int_{\BZ}^{\oplus} \dd k \, M_+(k) 
	. 
\end{align*}
All of these arguments can be straightforwardly adapted to include time-dependent media or lower-dimensional photonic crystals (\cf Section~\ref{Schroedinger:other_dimensions}). 

\subsection{The Bloch bundle and its topological classification} 
\label{classification:Bloch_bundle}
A vector bundle is a collection of vector spaces, indexed by a (base space) variable, that is glued together in a continuous or analytic fashion; for more information, we refer the interested reader to \cite[Section~3.3]{DeNittis_Lein:piezo_graphene:2013}, \cite[Section~IV.A]{DeNittis_Lein:exponentially_loc_Wannier:2011} and references therein. The Bloch bundle is a vector bundle constructed from a family of (energy or frequency) bands. When physicists use expressions such as “band topology”, what they actually mean is the following: they pick a family of isolated bands, which are relevant to the discussion.%
\begin{assumption}[Gap Condition]\label{classification:assumption:gap_condition}
	Suppose $\specrel(k) = \bigl \{ \omega_n(k) \bigr \}_{n \in \mathcal{I}}$ is a \emph{finite} family of relevant bands, indexed by a set of positive integers $\mathcal{I}$, that does not cross or merge with other bands. Put another way, they are separated by local spectral gaps from the other bands,
	\begin{align}
		\inf_{k \in \BZ} \mathrm{dist} \Bigl ( \specrel(k) \, , \, \sigma \bigl ( M_+(k) \bigr ) \setminus \specrel(k) \Bigr ) > 0 
		. 
		\label{classification:eqn:gap_condition}
	\end{align}
\end{assumption}
Once the relevant bands have been selected, the collection of vector spaces 
\begin{align*}
	\Hil_{\mathrm{rel}}(k) = \mathrm{span} \, \bigl \{ \varphi_n(k) \bigr \}_{n \in \mathcal{I}}
\end{align*}
which make up the Bloch bundle are the eigenspaces spanned by the relevant Bloch functions and indexed by Bloch momentum $k$. “Band topology” refers to how $\Hil_{\mathrm{rel}}(k)$ “twists and turns” as $k$ is varied; note that at this level the actual shape of the relevant frequency band functions $\omega_n(k)$ is irrelevant. 

Physically, the significance of the Gap Condition is that states supported in such an isolated family of bands decouple from the others, because band transitions outside of $\specrel(k)$ are typically exponentially suppressed. Note that the assumption that $\specrel(k)$ consists of \emph{finitely} many bands \emph{excludes ground state bands}, since they merge into the infinitely degenerate gradient field band $\omega_0(k) = 0$. This is not a mere technical obstacle either: since ground state Bloch waves for $k \approx 0$ have very long wavelengths, they no longer see the periodic structure but just \emph{homogeneous} material weights that are averaged over the unit cell. And homogeneous media require a different classification approach than periodic media. Mathematically, this manifests itself in the fact that the ground state Bloch functions are necessarily discontinuous, hence non-analytic, at $k = 0$, and this discontinuity prevents us from defining a vector bundle over the whole Brillouin torus $\BZ$. 

Symmetries of the Maxwell operator are inherited by the bundle: Should $M_+(k)$ possess a time-reveral symmetry $T$, for example, \ie $T$ is an antiunitary operator with $T \, M_+(k) = M_+(-k) \, T$, then $T$ relates the fibers $\Hil_{\mathrm{rel}}(k)$ and $\Hil_{\mathrm{rel}}(-k)$. Mathematically, symmetries of $M_+(k)$ give rise to additional structures on the vector bundle \cite{DeNittis_Gomi:AI_bundles:2014,DeNittis_Gomi:AII_bundles:2014,DeNittis_Gomi:AIII_bundles:2015}; this will be explained in Section~\ref{classification:Bloch_bundle:AI} below.

\subsubsection{The mathematical definition of the Bloch bundle} 
\label{classification:Bloch_bundle:definition}
One convenient way to think of $\Hil_{\mathrm{rel}}(k) = \ran P_{\mathrm{rel}}(k)$ is as the range of the projection 
\begin{align*}
	P_{\mathrm{rel}}(k) = \sum_{n \in \mathcal{I}} \sopro{\varphi_n(k)}{\varphi_n(k)} 
	, 
\end{align*}
although it is the \emph{weighted} scalar product that is implicit in the bra-ket notation, 
\begin{align}
	\sopro{\varphi_n(k)}{\varphi_n(k)} \psi(k) = \bscpro{\varphi_n(k)}{\psi(k)}_W \, \varphi_n(k)
	. 
	\label{classification:eqn:braket_weighted_scalar_product}
\end{align}
The Gap Condition ensures that $k \mapsto P_{\mathrm{rel}}(k)$ is not just locally analytic, but analytic over the whole Brillouin torus; put another way, locally around any point $k_0$ there is family of $k$-dependent unitaries $U(k,k_0)$ so that they are analytic and relate the projections at $k$ and $k_0$, $P_{\mathrm{rel}}(k) = U(k,k_0) \, P_{\mathrm{rel}}(k_0) \, U(k,k_0)^*$. Consequently, the dependence of $\Hil_{\mathrm{rel}}(k)$ on Bloch momentum $k$ is also analytic, and the Bloch bundle is the triple 
\begin{align}
	\mathcal{E}_{\BZ}(P_{\mathrm{rel}}) : \bigsqcup_{k \in \BZ} \Hil_{\mathrm{rel}}(k) \overset{\pi}{\longrightarrow} \BZ
	, 
\end{align}
consisting of the total space, the disjoint union of all the $\Hil_{\mathrm{rel}}(k)$s, the Brillouin torus $\BZ$ as base space, and the projection $\pi \bigl ( \Psi(k) \bigr ) = k$ onto the base point. Necessarily, the dimension $m = \dim \Hil_{\mathrm{rel}}(k)$ of the fiber vector space, the \emph{rank} of the vector bundle, has to be independent of $k$. In some contexts we need to distinguish between continuous and analytic vector bundles, although here, thanks to the so-called Oka principle, this is not necessary for vector bundles over the torus (the reader may find a detailed argument in \cite[Section~II.F]{DeNittis_Lein:exponentially_loc_Wannier:2011}). 

In the simplest case, the vector bundle is a \emph{trivial} complex vector bundle, meaning it is isomorphic to the \emph{product bundle} $\mathcal{E}_{\BZ}(P_{\mathrm{rel}}) \cong \BZ \times \C^m$ of base space and fiber. However, in general vector bundles can only be trivialized locally, \ie only in a sufficiently small neighborhood $\mathcal{U}$ of a point $k_0$ do we have $\pi^{-1}(\mathcal{U}) \cong \mathcal{U} \times \C^m$. In fact, a second and equivalent way to assemble the vector bundle from a trivializing open covering $\{ \mathcal{U}_j \}$ is to glue together $\pi^{-1}(\mathcal{U}_j) \cong \mathcal{U}_j \times \C^m$ using \emph{transition functions}, which then contain all the information on the “twists”. 

The Bloch vector bundle is suited to describe continuous deformations of the system: as $M_+(k)$ is deformed continuously, then also the frequency bands and the relevant subspaces $\Hil_{\mathrm{rel}}(k)$ change continuously as well — provided that the spectral gap does not close. Should additional symmetries be present, then these must not be broken during the deformation. Thus, continuous deformations of physical systems translate to continuous deformations of vector bundles. 

\subsubsection{Classification of complex vector bundles (CAZ class A)} 
\label{classification:Bloch_bundle:A}
To give rigorous meaning to the notion of “vector bundle up to continuous deformations”, we need to say when two vector bundles are considered equivalent and then \emph{classify equivalence classes of vector bundles}. This is quite standard for complex vector bundles and explicit criteria are known when the base space is low-dimensional ($d \leq 4$ suffices to cover time-dependent systems) and has such a simple structure as $\BZ \cong \T^d$ \cite{Nenciu:exponential_loc_Wannier:1983,Panati:triviality_Bloch_bundle:2006,DeNittis_Lein:exponentially_loc_Wannier:2011}. All of these are immediately relevant for our discussion of periodic light conductors, starting from layered media ($d = 1$) \cite{Choi_et_al:Zak_phase_quarter_wave_plates:2016}, two-dimensional \cite{van_Driel_et_al:tunable_2d_photonic_crystals:2004,Wang_et_al:edge_modes_photonic_crystal:2008,Wu_Hu:topological_photonic_crystal_with_time_reversal_symmetry:2015,Khanikaev_et_al:photonic_topological_insulators:2013} and three-dimensional \cite{Johnson_Joannopoulos:3d_photonic_crystal_band_gap:2000,Egen_et_al:PhCs_opals:2004,Joannopoulos_Johnson_Winn_Meade:photonic_crystals:2008,Kuramochi:fabrication_PhCs:2011} photonic crystals and the as-of-yet unrealized case of a three-dimensional photonic crystal which is deformed periodically in time.

\paragraph{Mathematical definition of equivalence} 
\label{par:mathematical_definition_of_equivalence}
Now let us explain when vector bundles are mathe\-matically equivalent: So let $\mathcal{E}_j = \bigl ( \xi_j , X , \pi_j \bigr )$, $j = 1 , 2$, be two vector bundles over the same base space. An $X$-map $f : \xi_1 \longrightarrow \xi_2$ is a continuous function between the total spaces so that the fiberwise restriction $f_x = f \vert_{\pi_1^{-1}(\{ x \})}$ defines a linear homomorphism between the vector spaces $\pi_1^{-1}(\{ x \})$ and $\pi_2^{-1}(\{ x \})$ attached to the same base point $x$. Put another way, $f$ preserves fibers and is compatible with the linear structure in each of the fibers. The set of such maps is denoted by $\mathrm{Hom}(\mathcal{E}_1,\mathcal{E}_2)$. If in addition $f$ restricts fiberwise to vector space \emph{iso}morphisms for all $x \in X$, then $f$ is in fact a homeomorphism between $\xi_1$ and $\xi_2$, and therefore defines an $X$-isomorphism between the bundles $\mathcal{E}_1$ and $\mathcal{E}_2$ \cite[Lemma~1.1]{Hatcher:vector_bundles_K_theory:2009}. This defines an equivalence relation $\mathcal{E}_1 \simeq \mathcal{E}_2$, and because isomorphic vector bundles have the same rank, we write $\mathrm{Vec}_{\C}^m(X)$ for the set of equivalence classes of isomorphic rank $m$ hermitian vector bundles. \emph{Classification theory of complex vector bundles concerns itself with the description of $\mathrm{Vec}^m_{\C}(X)$ for different $m$ and $X$.} One particularly important element is that associated to the \emph{trivial} vector bundle $\epsilon^m = \bigl ( X \times \C^m , X , \mathrm{proj}_1 \bigr )$ of rank $m$ where the total space is just the cartesian product of base space and fiber, and the projection $\mathrm{proj}_1(x,\psi) = x$ onto the first element. 

Mimicking the construction above, we could define equivalence of bundles in terms of \emph{analytic} $X$-isomorphisms, something that enters when establishing the existence of \emph{exponentially} localized Wannier functions \cite{Panati:triviality_Bloch_bundle:2006,Kuchment:exponential_decaying_wannier:2009,DeNittis_Lein:exponentially_loc_Wannier:2011}. Fortunately, though, in the present case we need not distinguish between continuous and analytic equivalence of vector bundles, because the so-called \emph{Oka principle} \cite{Grauert:analytische_Faserungen:1958} holds for $\BZ \cong \T^d$ (see \cite[Section~II.F]{DeNittis_Lein:exponentially_loc_Wannier:2011} for the detailed mathematical argument). 

\paragraph{Abstract classification results} 
\label{par:abstract_classification_results}
Now that we have defined $\mathrm{Vec}_{\C}^m(\T^d)$ as the set of vector bundles up to equivalence, two natural questions arise: first of all, \emph{how many} different equivalence classes are there? And secondly, given a concrete realization, can we \emph{compute} what equivalence class it belongs to? We postpone the second question and focus on the first. This classification problem is quite standard, and there are many different mathematical tools (\eg K-theory \cite{Husemoller:fiber_bundles:1966,Karoubi:K_theory:2008} or the vector bundle theoretic methods \cite{DeNittis_Gomi:AI_bundles:2014,DeNittis_Gomi:AII_bundles:2014,DeNittis_Gomi:AIII_bundles:2015}) with which we all arrive at the same conclusion: 
\begin{theorem}[Classification of $\mathrm{Vec}_{\C}^m(\T^d)$]\label{classification:thm:classification}
	For the cases of rank-$m$ vector bundles over the $d$-dimensional torus listed below, the set of equivalence classes is countable and given by: 
	\begin{enumerate}[(1)]
		\item $d = 1$, $m \geq 1$: $\mathrm{Vec}_{\C}^m(\mathbb{S}^1) \cong \{ 0 \}$
		\item $d \geq 2$, $m = 1$: $\mathrm{Vec}_{\C}^1(\T^d) \cong \Z$
		\item $d = 2$, $m \geq 2$: $\mathrm{Vec}_{\C}^m(\T^2) \cong \Z$
		\item $d = 3$, $m \geq 2$: $\mathrm{Vec}_{\C}^m(\T^3) \cong \Z^3$
		\item $d = 4$, $m \geq 2$: $\mathrm{Vec}_{\C}^m(\T^2) \cong \Z^6 \oplus \Z$
	\end{enumerate}
\end{theorem}
%

\paragraph{Computing topological invariants} 
\label{par:computing_topological_invariants}
The above classification would be quite academic if we were not able to actually compute the “coordinates” of a concretely given vector bundle in $\mathrm{Vec}_{\C}^m(\T^d)$; the coordinates which make up the address are \emph{topological invariants}, in this case first and second \emph{Chern classes}. To define the Chern classes 
\begin{align*}
	c_j(\mathcal{E}_{\BZ}) \in H^{2j}(\BZ,\Z) = \Z^{n(d,j)}
	, 
	&&
	j = 1 , 2 
	, 
\end{align*}
abstractly, one approach (see \eg \cite{Luke_Mischenko:vector_bundles_applications:1984,Milnor_Stasheff:characteristic_classes:1974}) is to view $\mathcal{E}_{\BZ}$ as the pullback of the universal vector bundle and define $\mathcal{E}_{\BZ}$'s Chern classes as the pullbacks of universal Chern classes. They are elements of the $2j$th cohomology group $H^{2j}(\BZ,\Z)$ over the Brillouin torus with integer coefficients; these cohomology groups can be computed explicitly by recursion to be $\Z$ to the power $n(d,j) = \nicefrac{d!}{j! \, (d-j)!}$ for all $0 \leq j \leq d$ and $n(d,j) = 0$ if $j > d$. 

Thanks to the Universal Coefficient Theorem \cite[Theorem~3.2]{Hatcher:algebraic_topology:2002}, we can embed $H^{2j}(\BZ,\Z)$ into the \emph{de Rham cohomology} $H^{2j}_{\mathrm{dR}}(\BZ)$, and Chern-Weil Theory \cite[Appendix~C]{Milnor_Stasheff:characteristic_classes:1974} allows us to connect the \emph{algebraically} defined Chern classes with \emph{differential geometric objects} such as the \emph{Berry curvature}. For details we refer the interested reader to \cite[pp.~28]{DeNittis_Lein:exponentially_loc_Wannier:2011}. 

The upshot is that these arguments not only ensure that they are \emph{integer}-valued, but also allow us to \emph{compute} Chern classes via the usual formulas. In case $d = 2$ the first Chern class can be identified with the first Chern number 
\begin{align}
	c_1(\mathcal{E}_{\BZ}) = \frac{1}{2\pi} \int_{\BZ} \dd k \, \mathrm{Tr}_{\Hil_+(k)} \, \bigl ( \Omega(k) \bigr ) 
	\label{classification:eqn:first_Chern_class}
\end{align}
whose equation involves the Berry \emph{curvature} $\Omega(k) = \dd \mathcal{A} + \mathcal{A} \wedge \mathcal{A}$, which in turn is defined \emph{locally} in terms of the Berry \emph{connection} $\mathcal{A}(k) = \bigl ( \mathcal{A}_{jn}(k) \bigr )_{1 \leq j , n \leq m}$
\begin{align*}
	\mathcal{A}_{jn}(k) = \ii \scpro{\varphi_j(k) \,}{\nabla_k \varphi_n(k)}_W
\end{align*}
with respect to a locally trivializing basis. Alternatively, we can recast 
\begin{align*}
	c_1(\mathcal{E}_{\BZ}) = \frac{1}{2\pi} \int_{\BZ} \dd k \, \mathrm{Tr}_{\Hil_+(k)} \Bigl ( \ii \, P_{\mathrm{rel}}(k) \; \bigl [ \partial_{k_1} P_{\mathrm{rel}}(k) \, , \, \partial_{k_2} P_{\mathrm{rel}}(k) \bigr ] \Bigr )
\end{align*}
solely in terms of the projection onto the relevant bands. Similar formulas exist for $c_1(\mathcal{E}_{\BZ})$ in $d \geq 3$ and for the second Chern class $c_2(\mathcal{E}_{\BZ})$ (\cf \eg \cite[equations~(5.9)–(5.10)]{DeNittis_Lein:exponentially_loc_Wannier:2011}). Collecting our results, we can restate the above classification Theorem~\ref{classification:thm:classification} as follows: 
\begin{corollary}\label{classification:cor:classification}
	The equivalence class a Bloch bundle belongs to is given in terms of the first two Chern classes. More specifically: 
	\begin{enumerate}[(1)]
		\item For the cases (2)–(4) enumerated in Theorem~\ref{classification:thm:classification} the Bloch vector bundle is classified by its first Chern class alone. 
		\item For case (5) ($d = 4$, $m \geq 2$), the Bloch vector bundle is classified by its first \emph{and} second Chern classes. 
	\end{enumerate}
	Put another way, two vector bundles can be continuously deformed into one another \emph{if and only if} all (first and second) Chern numbers agree; a vector bundle is \emph{trivial}, \ie isomorphic to the product bundle $\BZ \times \C^m$, if and only if all Chern numbers vanish. 
\end{corollary}
%

\subsubsection{Classification of vector bundles in the presence of even time-reversal symmetries (class AI)} 
\label{classification:Bloch_bundle:AI}
Since the work of Nenciu \cite{Nenciu:exponential_loc_Wannier:1983} (note also the later works \cite{Panati:triviality_Bloch_bundle:2006,DeNittis_Lein:exponentially_loc_Wannier:2011,Kuchment:exponential_decaying_wannier:2009}), it was understood that in $d \leq 3$ the presence of an anti-unitary $T$ satisfying 
\begin{align}
	T \, P_{\mathrm{rel}}(k) = P_{\mathrm{rel}}(-k) \, T
	\label{classification:eqn:triviality_condition_projection}
\end{align}
guarantees that the \emph{first} Chern numbers vanish. However, in general the presence of additional symmetries does \emph{not} mean the bundle is necessarily trivial, and the triviality of Chern numbers does not preclude the existence of \emph{other} topological invariants. In fact, if $T$ is odd, \ie $T^2 = - \id$ (class~AII), a new, $\Z_2$-valued topological invariant appears \cite{Furuta_et_al:Seibert-Witten_invariants:2000,Kane_Mele:Z2_ordering_spin_quantum_Hall_effect:2005,DeNittis_Gomi:AII_bundles:2014}. However, in the context of Maxwell operators time-reversal symmetries are necessarily even ($T^2 = + \id$) and have to satisfy \eqref{classification:eqn:triviality_condition_projection} (\cf Proposition~\ref{symmetries:prop:admissible_symmetries}). Consequently, if only one time-reversal symmetry is present, the classification theory for \emph{class~AI} vector bundles applies \cite{DeNittis_Gomi:AI_bundles:2014}. 
\begin{theorem}[{{{\cite[Theorem~1.6]{DeNittis_Gomi:AI_bundles:2014}}}}]\label{classification:thm:class_AI_bundles}
	Suppose there exists an anti-unitary operator $T$ satisfies \eqref{classification:eqn:triviality_condition_projection}. Then the first Chern class of the Bloch bundle $\mathcal{E}_{\BZ}(P_{\mathrm{rel}})$ vanishes and we have: 
	\begin{enumerate}[(1)]
		\item For $d \leq 3$ the Bloch bundle is trivial. 
		\item For $d = 4$ the Bloch bundle is trivial if and only if the second Chern class vanishes. 
	\end{enumerate}
\end{theorem}
Note, however, it is critical for this Theorem to hold that $T$ relates $P_{\mathrm{rel}}(k)$ and $P_{\mathrm{rel}}(-k)$. If $T \, P_{\mathrm{rel}}(k) = P_{\mathrm{rel}} \bigl ( \tau(k) \bigr ) \, T$ held for some other involution $\tau(k) \neq -k$, then the presence of the even time-reversal symmetry $T$ would \emph{not} necessarily ensure the triviality of the bundle \cite[Section~4]{DeNittis_Gomi:AI_bundles:2014}. 

\subsubsection{Classification of dual symmetric, non-gyrotropic materials} 
\label{classification:Bloch_bundle:dual_symmetric}
Dual symmetric, non-gyrotropic media have two even time-reversal symmetries and fall outside of the standard Cartan-Altland-Zirnbauer classification scheme. While media in $d \leq 3$ with a single even time-reversal symmetry are trivial, it is not permissible to simply omit the other. Moreover, because these two symmetries \emph{anti}commute with each other, their interplay needs to be carefully studied. We will perform the analysis of this topological class here. 

As luck would have it we can derive the classification with relatively simple, straight-forward arguments and do not need to perform a rather technical analysis along the lines of \cite{DeNittis_Gomi:AI_bundles:2014,DeNittis_Gomi:AII_bundles:2014,DeNittis_Gomi:AIII_bundles:2015} that involves advanced tools from the theory of vector bundles and $K$-theory. 

So we will have to perform the analysis ourselves here. The presence of two even time-reversal symmetries in such media, 
\begin{align*}
	T_j \, M_+(k) \, T_j^{-1} &= M_+(-k) 
	,
	\\
	T_j^2 &= + \id 
	, 
\end{align*}
where $j = 1 , 3$, means they automatically possess one unitary, commuting symmetry $U_2 = \sigma_2 \otimes \id = - \ii \, T_1 \, T_3$, 
\begin{align*}
	U_2 \, M_+(k) \, U_2^{-1} &= M_+(k) 
	, 
\end{align*}
which up to a factor of $- \ii$ equals the product of $T_1 = (\sigma_1 \otimes \id) \, C$ and $T_3 = ( \sigma_3 \otimes \id) \, C$. The presence of this symmetry means that each frequency band is \emph{helicity degenerate}, and we may choose Bloch functions with a specific (left- or right-handed) helicity. 

Put another way, the Maxwell operator 
\begin{align*}
	M_+ = \left (
	\begin{matrix}
		M_{+,+} & 0 \\
		0 & M_{+,-} \\
	\end{matrix}
	\right )
\end{align*}
admits a block decomposition into helicity components, where the operators $M_{+,\pm} = Q_{\pm} \, M_+ \, Q_{\pm}$ are obtained from the projections
\begin{align*}
	Q_{\pm} = \tfrac{1}{2} (\id \pm U_2) 
\end{align*}
onto right-handed (eigenvalue $+1$) and left-handed (eigenvalue $-1$) circularly polarized waves. 

Consequently, we may split the projection onto the relevant bands $P_{\mathrm{rel}}(k) = P_{\mathrm{rel},+}(k) + P_{\mathrm{rel},-}(k)$ as well as the Bloch bundle 
\begin{align*}
	\mathcal{E}_{\BZ}(P_{\mathrm{rel}}) = \mathcal{E}_{\BZ}(P_{\mathrm{rel},+}) \oplus \mathcal{E}_{\BZ}(P_{\mathrm{rel},-}) = \mathcal{E}_+ \oplus \mathcal{E}_-
\end{align*}
into a right-handed ($+$) and left-handed ($-$) component. Note that the ranks of $\mathcal{E}_{\pm}$ necessarily have to agree because each band has even helicity degeneracy. 

Below, we will show that the time-reversal symmetries are compatible with this decomposition (Lemma~\ref{classification:lem:block_diagonalization_symmetries}), independently of whether the time-reversal symmetries commute or anticommute (Lemma~\ref{classification:lem:dual_symmetric_commutating_symmetries}). That means they do not mix left- and right-handed states, and we may write $T_j = T_{j,+} \oplus T_{j,-}$ where $T_{j,\pm} = Q_{\pm} \, T_j \, Q_{\pm}$. What is more, the time-reversal symmetries restricted to each helicity component, $T_{1,\pm}$ and $T_{2,\pm}$, are no longer distinct physical symmetries (Theorem~\ref{classification:thm:block_decomposition_system}), and effectively, $\mathcal{E}_{\pm}$ come with only a \emph{single} even time-reversal symmetry. 

To summarize, for dual symmetric gyrotropic media the Bloch bundle is the sum of two class~AI bundles, and $\mathcal{E}_{\pm}$ can be classified with standard theory (see Theorem~\ref{classification:thm:class_AI_bundles}). 
\begin{theorem}[Classification of dual symmetric gyrotropic media]\label{classification:thm:dual_symmetric_materials}
	Suppose the weights $W$ are periodic with respect to $\Gamma \cong \Z^d$ and possess the two even time-reversal symmetries $T_1$ and $T_3$. Then the Bloch bundle associated to an isolated family of bands (in the sense of the Gap Condition~\ref{classification:assumption:gap_condition}) $\mathcal{E}_{\BZ}(P_{\mathrm{rel}}) = \mathcal{E}_+ \oplus \mathcal{E}_-$ is the sum of two class~AI bundles that can be classified as follows: 
	\begin{enumerate}[(1)]
		\item Independently of the dimension of the periodicity lattice, one topological invariant is the total rank of the bundle, $\mathrm{rank} \, \mathcal{E}_{\BZ} = 2 \, \mathrm{rank} \, \mathcal{E}_{\pm}$. 
		\item When the rank is fixed, then in dimensions $d = 1 , 2 , 3$ the Bloch bundle is the sum of two trivial vector bundles, \ie all vector bundles are topologically equivalent. 
		\item When the rank is fixed, then in dimensions $d = 4$ the Bloch vector bundles are distinguished by the two second Chern numbers of $\mathcal{E}_{\pm}$. 
	\end{enumerate}
\end{theorem}
We emphasize that the notion of trivial vector bundle depends on the symmetries, \ie on the \emph{topological class}: for complex vector bundles (the quantum Hall class, class~A, with no symmetries) a trivial bundle by definition is one that is isomorphic to a product bundle. Given that the Chern numbers for product bundles all vanish, this gives a simple criterion for bundle triviality. For vector bundles that are endowed with additional symmetries, the situation can be more delicate. For low-dimensional class~AI bundles over the torus there is only one phase and all vector bundles of the same rank can be continuously deformed into one another (provided the time-reversal symmetry is preserved); put another way, not just a specific bundle, but the whole class of bundles is trivial. However, for other topological classes (\eg class~AIII) where bundles are characterized by \emph{relative} topological invariants, there exists no canonical notion of trivial bundle. 
\medskip

\noindent
We now proceed with the derivation of this classification result; our arguments are divided up into several steps. The proofs are all non-technical and we think they are instructive for the reader to follow our reasoning with which we arrive at Theorem~\ref{classification:thm:dual_symmetric_materials}. We will broaden our setting: suppose $H$ is a selfadjoint (hermitian) operator, a stand-in for a quantum Hamiltonian or the Maxwell operator, that comes furnished with two antiunitary symmetries, 
\begin{subequations}\label{classification:eqn:epsilon_lambda_type_symmetry_conditions}
	\begin{align}
		T_j \, H \, T_j^{-1} &= \epsilon_j H 
		, 
		&&
		j = 1 , 2 
		, 
		\label{classification:eqn:epsilon_lambda_type_symmetry_conditions:TR_vs_PH}
		\\
		T_j^2 &= \lambda_j \id 
		. 
		\label{classification:eqn:epsilon_lambda_type_symmetry_conditions:even_odd}
	\end{align}
\end{subequations}
That means we do not restrict ourselves to even time-reversal symmetries, but admit any combination of even ($\lambda = 1$) or odd ($\lambda = -1$) symmetries, be it of time-reversal- ($\epsilon = + 1$) or of particle-hole-type ($\epsilon = -1$). For the sake of brevity, we will say that $T_j$ is of type $(\epsilon_j,\lambda_j)$. 

First, we have a certain amount of freedom when picking antiunitary symmetries, for example we may multiply them with a phase to transform an \emph{anti}commuting to a commuting pair of symmetries. 
\begin{lemma}\label{classification:lem:equivalence_symmetries}
	Suppose $T$ is a symmetry of type $(\epsilon,\lambda)$ in the sense of equation~\eqref{classification:eqn:epsilon_lambda_type_symmetry_conditions} and $\e^{\ii \varphi} \in \C$ a phase. Then $T' = \e^{\ii \varphi} \, T$ is a symmetry of type $(\epsilon,\lambda)$. 
\end{lemma}
\begin{proof}
	The antiunitarity of $T$ implies ${T'}^2 = \e^{\ii \varphi} \, \e^{- \ii \varphi} \, T^2 = \lambda \, \id$, and $T'$ is even or odd whenever $T$ is. 
	
	To show that $T'$ (anti)commutes with $H$ if and only if $T$ does, we express $T = U \, C$ as the product of a unitary $U$ and complex conjucation $C$ (which is always possible). Therefore, the inverse of $T'$, 
	\begin{align*}
		{T'}^{-1} &= \bigl ( \e^{\ii \varphi} \, U \, C \bigr )^{-1} 
		= C \, \bigl ( \e^{\ii \varphi} \, U \bigr )^* 
		= C \, \e^{- \ii \varphi} \, U^*
		\\
		&= \e^{+ \ii \varphi} \, C \, U^* 
		= \e^{\ii \varphi} \, T^{-1} 
		, 
	\end{align*}
	is just $\e^{\ii \varphi}$ times the inverse of $T$, and hence, $T'$ is a symmetry of the same type, 
	\begin{align*}
		T' \, H \, {T'}^{-1} &= T \, \e^{- \ii \varphi} \, H \, \e^{\ii \varphi} \, T^{-1} 
		= T \, H \, T^{-1} 
		= \epsilon \, H 
		. 
	\end{align*}
\end{proof}
Secondly, if the two antiunitary symmetries commute up to a phase, 
\begin{align}
	T_1 \, T_2 = \e^{\ii \varphi} \, T_2 \, T_1 
	, 
	\label{classification:eqn:commutation_up_to_phase}
\end{align}
then we can find equivalent symmetries which commute; here, the previous Lemma . In particular, two anticommuting, antiunitary symmetries are equivalent to a pair of antiunitary, \emph{commuting} symmetries. 
\begin{lemma}\label{classification:lem:dual_symmetric_commutating_symmetries}
	Suppose $T_1$ and $T_2$ are two antiunitary symmetries of types $(\epsilon_j , \lambda_j)$, $j = 1 , 2$, which commute up to a phase $\e^{\ii \varphi} \in \C$. Then $T_1' = T_1$ and $T_2' = \e^{\ii \frac{\varphi}{2}} \, T_2$ are a pair of \emph{commuting} symmetries of the same types. 
\end{lemma}
\begin{proof}
	According to the preceding Lemma, $T_2'$ is also of type $(\epsilon_2 , \lambda_2)$, the same as $T_2$. And a quick computation shows that $T_2'$ commutes with $T_1' = T_1$ because the phase factors cancel: 
	\begin{align*}
		T_1' \, T_2' &= T_1 \, \e^{\ii \frac{\varphi}{2}} \, T_2 
		= \e^{- \ii \frac{\varphi}{2}} \, T_1 \, T_2 
		= \e^{- \ii \frac{\varphi}{2}} \, \e^{\ii \varphi} \, T_2 \, T_1 
		= T_2' \, T_1' 
	\end{align*}
\end{proof}
Now suppose the two symmetries commute and are of the same type, $(\epsilon_1,\lambda_1) = (\epsilon_2,\lambda_2)$. As we have just seen, assuming $T_1 \, T_2 = T_2 \, T_1$ instead of \eqref{classification:eqn:commutation_up_to_phase} imposes no additional restrictions. Then their product 
\begin{align*}
	U = T_1 \, T_2 
\end{align*}
is a unitary that squares to $+ \id$, 
\begin{align}
	U^2 &= \bigl ( T_1 \, T_2 \bigr )^2 
	= T_1^2 \, T_2^2 
	= \lambda^2 \, \id 
	= + \id 
	, 
	\label{classification:eqn:V_squares_to_1}
\end{align}
and commutes with $H$, 
\begin{align}
	U \, H \, U^{-1} &= T_1 \, \bigl ( T_2 \, H \, T_2^{-1} \bigr ) \, T_1^{-1} 
	= \epsilon \, T_1 \, H \, T_1^{-1} 
	= \epsilon^2 \, H 
	= H 
	. 
	\label{classification:eqn:V_commutes_with_H}
\end{align}
Combining the unitarity of $U$ with \eqref{classification:eqn:V_squares_to_1}, we deduce that $U = U^* = U^{-1}$ is selfadjoint (hermitian), and due to $U^2 = + \id$ the spectrum $\sigma(U) \subseteq \{ -1 , +1 \}$ consists of either one or two eigenvalues. Let us exclude the trivial cases $U = \pm \id$ where $T_1$ and $T_2 = \pm T_1^{-1} = \pm \lambda T_1$ are equivalent as they differ only by a phase (Lemma~\ref{classification:lem:equivalence_symmetries}). When $U \neq \pm 1$ both eigenvalues occur and the corresponding spectral projections 
\begin{align*}
	Q_{\pm} = \tfrac{1}{2} (\id \pm U) 
\end{align*}
map onto the eigenspaces associated with the eigenvalues $\pm 1$. For dual symmetric electromagnetic media these are the projections onto right-handed ($+1$) and left-handed ($-1$) electromagnetic waves. 

These give rise to a decomposition of the Hilbert space 
\begin{align*}
	\Hil = \Hil_+ \oplus \Hil_- 
\end{align*}
where the two subspaces $\Hil_{\pm} = Q_{\pm}[\Hil]$ are the ranges of $Q_{\pm}$. As $U$ commutes with $H$ (equation~\eqref{classification:eqn:V_commutes_with_H}), so do the associated spectral projections, which leads us to conclude 
\begin{align*}
	H = H_+ \oplus H_- = \left (
	\begin{matrix}
		H_+ & 0 \\
		0 & H_- \\
	\end{matrix}
	\right )
	, 
\end{align*}
since $Q_{\pm} \, H \, Q_{\mp} = H \, Q_{\pm} \, Q_{\mp} = 0$ and only the block-diagonal contributions $H_{\pm} = Q_{\pm} \, H \, Q_{\pm}$ remain. It turns out that also the two antiunitary symmetries $T_j = T_{j,+} \oplus T_{j,-}$ are block-diagonal with respect to this decomposition. 
\begin{lemma}\label{classification:lem:block_diagonalization_symmetries}
	Suppose $T_1$ and $T_2$ are two commuting, antiunitary symmetries of the same type $(\epsilon,\lambda)$, and their product $U \neq \pm \id$ is not trivial. Then $T_1$, $T_2$ and $U$ are block-diagonal, 
	\begin{align*}
		T_j &= \left (
		\begin{matrix}
			T_{j,+} & 0 \\
			0 & T_{j,-} \\
		\end{matrix}
		\right )
		, 
		&&
		j = 1 , 2 
		, 
		\\
		U &= \left (
		\begin{matrix}
			+ \id & 0 \\
			0 & - \id \\
		\end{matrix}
		\right )
		. 
	\end{align*}
	Moreover, the block components of the two antiunitaries are related, 
	\begin{align}
		T_{2,\pm} &= \pm T_{1,\pm}^{-1} 
		= \pm \lambda T_{1,\pm} 
		. 
		\label{classification:eqn:relation_U_1_U_2_componentwise}
	\end{align}
\end{lemma}
The condition $U \neq \pm \id$ is necessary to ensure that the two time-reversal symmetries are \emph{distinct} and the block decomposition meaningful. 
\begin{proof}
	Since $T_1$ and $T_2$ commute with one another by assumption, $U$ also necessarily commutes with $T_j$, 
	\begin{align*}
		U \, T_j &= T_1 \, T_2 \, T_j 
		= T_j \, T_1 \, T_2 
		= T_j \, U
		. 
	\end{align*}
	Consequently, they commute with the spectral projections $Q_{\pm} = \tfrac{1}{2} (\id \pm U)$ as well, and the symmetries are all block-diagonal with respect to $\Hil = \Hil_+ \oplus \Hil_-$. 
	
	Computing $U$ block-wise yields 
	\begin{align*}
		U &= (+\id) \oplus (-\id) 
		= \bigl ( T_{1,+} \, T_{2,+} \bigr ) \oplus \bigl ( T_{1,-} \, T_{2,-} \bigr ) 
		, 
	\end{align*}
	and comparing left- and right-hand side yields equation~\eqref{classification:eqn:relation_U_1_U_2_componentwise}. 	
\end{proof}
The main result is now within reach: 
\begin{theorem}\label{classification:thm:block_decomposition_system}
	Suppose $T_1$ and $T_2$ are two antiunitary symmetries of type $(\epsilon,\lambda)$ that commute up to the phase $\e^{\ii \varphi}$ according to equation~\eqref{classification:eqn:commutation_up_to_phase}. Then the system block decomposes with respect to $\Hil = \Hil_+ \oplus \Hil_-$ (defined as above), namely $H = H_+ \oplus H_-$ and $T_j = T_{j,+} \oplus T_{j,-}$ for $j = 1 , 2$. Moreover, on $\Hil_{\pm}$ one of the symmetries is redundant, because 
	\begin{align*}
		T_{2,\pm} &= \pm \e^{- \ii \frac{\varphi}{2}} \, T_{1,\pm}^{-1} = \pm \lambda \e^{- \ii \frac{\varphi}{2}} \, T_{1,\pm} 
	\end{align*}
	equals $T_{1,\pm}$ up to a phase. 
\end{theorem}
\begin{proof}
	For commuting symmetries where $\e^{\ii \varphi} = 1$ the statement is an immediate consequence of Lemma~\ref{classification:lem:block_diagonalization_symmetries} and Lemma~\ref{classification:lem:equivalence_symmetries}. 
	
	When $\e^{\ii \varphi} \neq 1$ we apply the above argument to $T_1' = T_1$ and $T_2' = \e^{\ii \frac{\varphi}{2}} \, T_2$ as they are equivalent \emph{commuting} symmetries (Lemma~\ref{classification:lem:dual_symmetric_commutating_symmetries}). To translate that back to the generic case, all we need to do is add the phase factor $\e^{- \ii \frac{\varphi}{2}}$, 
	\begin{align*}
		T_2 &= \e^{- \ii \frac{\varphi}{2}} \, T_2' 
		= \bigl ( + \lambda \e^{- \ii \frac{\varphi}{2}} \, T_{1,+} \bigr ) \oplus \bigl ( - \lambda \e^{- \ii \frac{\varphi}{2}} \, T_{1,-} \bigr )
		. 
	\end{align*}
\end{proof}
The classification of Bloch bundles with two even time-reversal symmetries is just a special case: 
\begin{corollary}\label{classification:cor:classification_dual_time_symmetry}
	Suppose $M_+$ is the Maxwell operator for a periodic, dual symmetric gyrotropic medium, and $P_{\mathrm{rel}}(k)$ the projection associated to a family of isolated frequency bands (\ie they satisfy the Gap Condition~\ref{classification:assumption:gap_condition}). Then the associated Bloch bundle 
	\begin{align}
		\mathcal{E}_{\BZ}(P_{\mathrm{rel}}) = \mathcal{E}_+ \oplus \mathcal{E}_- 
		\label{classification:eqn:Bloch_bundle_Maxwell_decomposition}
	\end{align}
	decomposes into the sum of two class~AI vector bundles $\mathcal{E}_{\pm} = \mathcal{E}_{\BZ}(P_{\mathrm{rel},\pm})$, each associated to one helicity component. 
\end{corollary}
\begin{proof}
	First of all, $M_+$ is a selfadjoint operator furnished with two even time-reversal symmetries that anticommute (\ie the phase is $-1 = \e^{\ii \pi}$). Hence, we are in the setting of Theorem~\ref{classification:thm:block_decomposition_system}. Instead of $T_3 = (\sigma_3 \otimes \id) \, C$ we may use $\ii T_3$ so that it commutes with $T_1 = (\sigma_1 \otimes \id) \, C$ and their product yields $U = T_1 \, \ii T_3 = \sigma_2 \otimes \id$. The eigenspaces of $U$ correspond to electromagnetic fields with right- ($+1$) and left-handed ($-1$) helicities. Therefore, $M_+ = M_{+,+} \oplus M_{+,-}$ and $P_{\mathrm{rel}}(k) = P_{\mathrm{rel},+}(k) \oplus P_{\mathrm{rel},-}(k)$ as well as the two time-reversal symmetries decompose into helicity components. The components of the projection $P_{\mathrm{rel},\pm}(k) = Q_{\pm} \, P_{\mathrm{rel}}(k) \, Q_{\pm}$ are defined analogously to $H_{\pm}$; the projections $Q_{\pm} = \tfrac{1}{2} (\id + U)$ do not depend on crystal momentum because $U$ is independent of $x$ and $- \ii \nabla$. 
	
	Setting $\mathcal{E}_{\pm} = \mathcal{E}_{\BZ}(P_{\mathrm{rel},\pm})$, we see that the Bloch bundle splits into helicity components as given in equation~\eqref{classification:eqn:Bloch_bundle_Maxwell_decomposition}. On each helicity component, the two symmetries $T_{1,\pm}$ and $T_{3,\pm} = \pm \e^{- \ii \frac{\pi}{2}} \, T_{1,\pm}$ are equivalent. Thus, the vector bundles $\mathcal{E}_{\pm}$ are endowed with only \emph{one} even time-reversal symmetry, and we may consider them as class~AI vector bundles. 
\end{proof}
Note that analogous statements hold for when $H$ has two odd time-reversal or two (even or odd) particle-hole symmetries. In those cases, the Bloch bundle splits into two class~AII, class~D or class~C bundles, respectively. However, these cases are not relevant in our analysis of photonic crystals. 
\begin{proof}[Theorem~\ref{classification:thm:dual_symmetric_materials}]
	\begin{enumerate}[(1)]
		\item The fact that the frequency bands for right- and left-handed circularly polarized Bloch waves always have the same multiplicity stems from the fact that the degeneracy of $+1$ and $-1$ of $\sigma_n$ are the same. Therefore, the ranks of the right-handed and left-handed circularly polarized sub bundles $\mathcal{E}_{\pm}$ are always the same. 
		\item Corollary~\ref{classification:cor:classification_dual_time_symmetry} tells us that the Bloch bundle can be seen as the sum of two class~AI bundles. In dimension $d = 1 , 2 , 3$ vector bundles over the torus where the time-reversal symmetry relates fibers at $\pm k$ are trivial (Theorem~\ref{classification:thm:class_AI_bundles}). 
		\item This again follows from Corollary~\ref{classification:cor:classification_dual_time_symmetry} and the fact that class~AI bundles are distinguished topologically by the second Chern number of which there are two (Theorem~\ref{classification:thm:class_AI_bundles}). 
	\end{enumerate}
\end{proof}
%
\section{Conclusion and future developments} 
\label{conclusion}
We close this work by contrasting and comparing our results to the literature and sketch what avenues we would like to explore in the future. Our main aim in this paper was to make precise what “topological photonic crystal” means, and how to differentiate between topologically distinct types of electromagnetic media. Our first principles approach is very different from most attempts in the literature where the keyword “topological” is tacked onto a lot of physical effects even if the link to topology is not made explicit. We emphasize that phenomenological similarities such as a locking of spin and propagation direction are \emph{not conclusive evidence of a topological origin.} Instead, a direct causal link should be established between physical effects and the topology of a mathematical object.

\subsection{Topological effects due to material symmetries} 
\label{conclusion:summary}
We framed this article with two specific questions, and we owe it to the reader to provide the answers we have promised in the introduction. The first concerns the similarities between the Quantum Hall Effect in photonic crystals and condensed matter physics.

\subsubsection{Haldane's Quantum Hall Effect for light} 
\label{conclusion:comparison:QHE}
We can only give a partial answer: Gyrotropic media belong to the same topological class as quantum systems exhibiting the Quantum Hall Effect — in both cases an even time-reversal symmetry is broken so that they both belong to \emph{class~A}. That is consistent with Haldane's conjecture, and we intend to give a complete derivation in a future work. 

Our analysis provides a number of new insights: first of all, the relevant symmetry that is broken is the even time-reversal symmetry $T_3$ rather than complex conjugation $C$ as is argued in some of the literature (see \eg \cite{Raghu_Haldane:quantum_Hall_effect_photonic_crystals:2008,Wu_Hu:topological_photonic_crystal_with_time_reversal_symmetry:2015}). In fact, complex conjugation symmetry of Maxwell's equations can \emph{never} be broken since one of the tenets of classical electromagnetism is that fields $(\mathbf{E},\mathbf{H})$ are necessarily real; this reality constraint is preserved in the correct complexified Maxwell equations (\cf \cite[Section~2.2]{DeNittis_Lein:ray_optics_photonic_crystals:2014}) where \emph{$C$ acts as an even particle-hole symmetry}. Such subtle distinctions are essential for a topological classification, because it is crucial we correctly identify the nature of the relevant symmetries. 

The bulk classification of gyrotropic photonic crystals as class~A topological insulators made here is an important first step towards establishing Haldane's Photonic Bulk-Edge Correspondence conjecture. At first glance, adapting derivations of the Quantum Hall Effect may now seem completely straightforward and further research unnecessary: there is a wealth of literature (\eg \cite{Thouless_Kohmoto_Nightingale_Den_Nijs:quantized_hall_conductance:1982,Bellissard_van_Elst_Schulz_Baldes:noncommutative_geometry_quantum_hall_effect:1994,Kellendonk_Richter_Schulz-Baldes:edge_currents_Chern_numbers_quantum_Hall:2002,Kellendonk_Schulz-Baldes:quantization_edge_currents:2004,Kellendonk_Schulz-Baldes:boundary_maps_Cstar_quantum_Hall:2004,Prodan_Schulz_Baldes:complex_topological_insulators:2016}) with different approaches to deriving bulk-edge correspondences that generalize the early works of Hatsugai \cite{Hatsugai:Chern_number_edge_states:1993,Hatsugai:edge_states_Riemann_surface:1993}. 

However, there are important mathematical \emph{and} physical differences between the quantum system and its electromagnetic analog, and these differences are not mere technical footnotes but essential. Taking these differences into account was already important to properly understand the analogy between semiclassical limits and the derivation of ray optics equations \cite[Section~5]{DeNittis_Lein:ray_optics_photonic_crystals:2014}. One of the issues that was discussed there was the \emph{form of “typical” states}: in solid state physics, the relevant states are perturbations of the Fermi projection where all states up to the Fermi energy $E_{\mathrm{F}}$ are completely filled; in the context of topological insulators, the Fermi energy is usually assumed to lie in a spectral gap or, more generally, a zone of dynamical localization. Electromagnetic waves, though, are not fermions, so there is no physical principle that forbids us to arbitrarily populate bands. Instead, states are often excited by a laser, and therefore, states are peaked around some wave vector $k_0$ and frequency $\omega_0$, \ie wave packet states. Another method is to use an antenna to excite a given frequency in multiple directions at the same time. Nevertheless, it stands to reason that the analog of the “Fermi projection” may enter the \emph{derivation} of this photonic bulk-edge correspondence as an \emph{auxiliary object} (as opposed to being interpreted as the physical state of the system). 

A second difference is that there does not seem to be a photonic “bulk observable” and only the edge observable “net number of boundary states” (right moving vs.\ left-moving) enters equations~\eqref{intro:eqn:potential_symmetries} so that we need to prove the following two equalities 
\begin{align*}
	\mbox{signed $\sharp$ edge modes} = C_{1,\mathrm{edge}} = C_{1,\mathrm{bulk}} 
	. 
\end{align*}
A prerequisite for proving this photonic bulk-edge correspondence is to be able to compute the Chern number of the “photonic Fermi projection” (again, seen only as an auxiliary object). Unfortunately, our definition of Bloch bundle from Section~\ref{classification:Bloch_bundle} does not apply without modification. We had to exclude the so-called ground state bands, \ie the bands which have approximately linear dispersion near $k = 0$ and $\omega = 0$; these necessarily exist in \emph{any} periodic medium \cite[Theorem~1.4~(iii)]{DeNittis_Lein:adiabatic_periodic_Maxwell_PsiDO:2013}. If the relevant bands include the ground state bands, then the dimension $\mathrm{dim} \, \Hil_{\mathrm{rel}}(0) = \mathrm{dim} \, \Hil_{\mathrm{rel}}(k) + 2$ is larger by $2$ than for $k \neq 0$  as the transversality condition  
\begin{align*}
	\left (
	\begin{matrix}
		(- \ii \nabla + k) \cdot \psi^E(k) \\
		(- \ii \nabla + k) \cdot \psi^H(k) \\
	\end{matrix}
	\right ) = 
	\left (
	\begin{matrix}
		0 \\
		0 \\
	\end{matrix}
	\right )
\end{align*}
degenerates there; this is further explained in \cite[Sections~3.2–3.3]{DeNittis_Lein:adiabatic_periodic_Maxwell_PsiDO:2013}. We intend to revisit this question in a future work. 

A third, perhaps more mathematical question concerns how to create interfaces (edges and surfaces) in the first place, and one can choose from at least three different options: 
\begin{enumerate}[(1)]
	\item We can sandwich two different photonic crystals \cite{Khanikaev_et_al:photonic_topological_insulators:2013,Wu_Hu:topological_photonic_crystal_with_time_reversal_symmetry:2015} whose relevant photonic band gaps overlap. 
	\item We can terminate the photonic crystal with a metal \cite{Wang_et_al:unidirectional_backscattering_photonic_crystal:2009}; the metal is typically modeled as a perfect electric conductor and enters as boundary conditions imposed on solutions to Maxwell's equations. 
	\item In principle, we could also choose to terminate with a perfect \emph{magnetic} conductor, which translates to a different set of boundary conditions. 
\end{enumerate}
While in spirit this closely resembles the quantum case, we may obtain \emph{different} photonic bulk-edge correspondences for some or all of these choices. Such differences already appear when considering interfaces between \emph{homogeneous} (as opposed to periodic) electromagnetic media, where the existence and polarization of surface modes depends on the details \cite[Supplementary Material]{Bliokh_Smirnova_Nori:spin_orbit_light:2015}. Therefore, it is not a foregone conclusion that the form of the photonic bulk-edge correspondences is independent of how boundaries are created. 

\subsubsection{Absence of novel topological effects in $d \leq 3$ due to material symmetries} 
\label{conclusion:summary:absence_of_novel_topological_effects}
The second question concerned the existence of novel topological effects, \ie other than the Quantum Hall Effect for light. When we began our investigation, we were hoping to find that certain media belong to topological classes that support topological invariants other than Chern numbers, which in turn would indicate the existence of other topological effects. Unfortunately, this is not the case: barring additional (isospin) symmetries, our finding here is that apart from gyrotropic materials, which are in the same topological class as the Quantum Hall Effect, the other three types of media do not support as-of-yet unknown topological effects in $d \leq 3$; for three-dimensional photonic crystals with periodic time-dependence, there could be topological effects related to the second Chern number — independently of which of the four topological classes the material belongs to. While that does not exclude topological effects due to crystallographic symmetries (\eg \cite{Choi_et_al:Zak_phase_quarter_wave_plates:2016,Fang_Gilbert_Bernevig:topological_insulators_point_group_symmetries:2012,Alexandradinata_Fang_Gilbert_Bernevig:spin_orbit_free_topological_insulators_without_TRS:2014,Chen_et_al:topological_photonic_crystals_crystallographic_symmetries:2015}), those will have to be considered separately. The case of time-reversal symmetric media with \emph{additional} (linear, commuting) symmetries reduces either to class~AI (topologically trivial) or the Quantum Hall Class, class~A; we will discuss both of these in more detail next. 

\subsection{Comparison with the literature} 
\label{conclusion:comparison}
Ever since Raghu and Haldane's first proposed topological phenomena in periodic electromagnetic media \cite{Raghu_Haldane:quantum_Hall_effect_photonic_crystals:2008}, there has been a growing body of work on the subject; for recent reviews we point to \cite{Lu_Joannopoulos_Soljacic:topological_photonics:2014,Lu_Joannopoulos_Soljacic:topological_photonics:2016}. To keep the discussion brief, we will focus on a few select publications that are directly related to the core of this work and representative for a number of others.

\subsubsection{Haldane's Quantum Hall Effect for light} 
\label{conclusion:comparison:QHE}
While the experimental confirmation \cite{Wang_et_al:unidirectional_backscattering_photonic_crystal:2009} of Quantum Hall Effect for light settled the question that topological phenomena exist, very little effort was made to probe the \emph{quantitative} validity of Haldane's Photonic Bulk-Edge Correspondence and derive it from first principles.

\paragraph{Ray optics equations} 
\label{par:ray_optics_equations}
Raghu and Haldane based their arguments on \emph{postulating} ray optics equations which included an “anomalous velocity term”; however, the form of the sub-leading terms was a topic of discussion, one that was settled only recently with our rigorous work \cite{DeNittis_Lein:ray_optics_photonic_crystals:2014} (see \cite[Section~5.2]{DeNittis_Lein:ray_optics_photonic_crystals:2014} for an in-depth discussion). Unfortunately, it is not possible to derive bulk-boundary correspondences purely on the basis of ray optics equations — not only because those govern the light dynamics in \emph{the bulk}, but also because the semiclassical arguments with which one may show the quantization of the transverse conductivity (see \eg \cite[Section~1]{PST:effective_dynamics_Bloch:2003}) do \emph{not} generalize to electromagnetism due to the fundamental differences between both physical theories (\cf the discussion in \cite[Section~5.1]{DeNittis_Lein:ray_optics_photonic_crystals:2014}). 

\paragraph{Justification for effective “Hamiltonians”} 
\label{par:justification_for_effective_hamiltonians}
A great deal of theoretical works on topological phenomena in photonic crystals argue in two steps: first, a system is identified which has the desired features in its frequency band diagram (in the simplest case a photonic band or a conical intersection). Then based on the dispersion of the frequency bands of interest, a simpler, effective Hamiltonian is \emph{postulated} that shares the same essential features in its band spectrum. This \emph{ad hoc} procedure is performed without making any reference to the dynamical problem – which has been falsely considered to be well-understood by analogy to the Bloch electron’s quantum dynamics. 

The reason for this is that in a great number of cases, physicists study the \emph{second}-order Maxwell equations obtained by “squaring” equation~\eqref{Schroedinger:eqn:Maxwell_equations}: that is because in media where the bianisotropic tensor $\chi = 0$ vanishes, the second-order eigenvalue problem for electric and magnetic field can be solved separately, \eg 
\begin{align*}
	\eps^{-1} \, \nabla \times \bigl ( \mu^{-1} \nabla \times \varphi_n^E(k) \bigr ) = \omega_n(k)^2 \, \varphi_n^E(k) 
\end{align*}
for the electric Bloch functions and an analogous equation for the magnetic Bloch functions.\footnote{We will ignore the problem of properly defining this equation when $\epsilon$ or $\mu$ is hermitian instead of real-symmetric; because we lose the information on the sign and we are no longer able to implement the restriction to positive frequencies in a straightforward fashion.} The “effective Hamiltonian” then supposedly approximates the physics of the full equations associated to $M^2_{EE} = \eps^{-1} \, \nabla^{\times} \, \mu^{-1} \nabla^{\times}$ (see \eg \cite{Wu_Hu:topological_photonic_crystal_with_time_reversal_symmetry:2015}). 

However, things are not as simple as they appear. In the absence of sources the dynamical equation for the electric field is the wave equation 
\begin{align*}
	\partial_t^2 \psi^E(t) + M^2_{EE} \psi^E(t) &= 0 
	, 
\end{align*}
and as it is second order in time, we not only need to specify $\psi^E(t_0)$ but also the time derivative $\partial_t \psi^E(t_0) = \eps^{-1} \, \nabla \times \psi^H(t_0)$, which is determined through the \emph{magnetic} field at initial time $t_0$. Consequently, the solution \emph{cannot} be of the form $\e^{- \ii (t - t_0) M^2_{EE}} \psi^E(t_0)$, which would allow us to replace the “evolution group” $\e^{- \ii t M^2_{EE}}$ by the effective “evolution” $\e^{- \ii t H_{\mathrm{eff}}}$. 

Again, this problem disappears if we stick to the proper Schrödinger formalism where the evolution equation is first-order in time and is mathematically of the form of a Schrödinger equation, namely 
\begin{align*}
	\ii \partial_t \Psi(t) = M \Psi(t)
	,
	&& 
	\Psi(t_0) = \Phi
	,
\end{align*}
where $M = W^{-1} \, \Rot \big \vert_{\omega \geq 0}$ is selfadjoint (hermitian). Now suppose we are given a closed subspace $\Hil_{\mathrm{rel}}$ spanned by states which we deem relevant, and that this subspace is left invariant by the dynamics; one common example would be states from a given finite frequency range. That means if we start with $\Phi \in \Hil_{\mathrm{rel}}$ then the time-evolved state $\Psi(t) = \e^{- \ii t M} \Phi \in \Hil_{\mathrm{rel}}$ remains in the relevant subspace. This translates to $M \, P_{\mathrm{rel}} = P \, M_{\mathrm{rel}}$ for the orthogonal projection $P_{\mathrm{rel}}$ onto $\Hil_{\mathrm{rel}}$. Now if you can approximate $M \, P_{\mathrm{rel}} \approx M_{\mathrm{eff}} \, P_{\mathrm{rel}}$ by some effective Maxwell operator $M_{\mathrm{rel}}$ for states from $\Hil_{\mathrm{rel}}$, then a Duhamel argument 
\begin{align}
	\e^{- \ii t M} \, P_{\mathrm{eff}} - \e^{- \ii t M_{\mathrm{eff}}} \, P_{\mathrm{eff}} &= 
	\int_0^t \dd s \, \frac{\dd }{\dd s} \Bigl ( \e^{- \ii t M} \, \e^{- \ii (t - s) M_{\mathrm{eff}}} \Bigr ) \, P_{\mathrm{eff}} 
	\notag \\
	&\approx \int_0^t \dd s \,  \e^{- \ii t M} \, \bigl ( M - M_{\mathrm{eff}} \bigr ) \, P_{\mathrm{eff}} \, \e^{- \ii (t - s) M_{\mathrm{eff}}}
	\approx 0 
	\label{conclusion:eqn:Duhamel_argument}
\end{align}
yields that the effective dynamics remain close to the full dynamics, $\e^{- \ii t M_{\mathrm{eff}}} \, P_{\mathrm{eff}} \approx \e^{- \ii t M} \, P_{\mathrm{eff}}$. Evidently, this argument (which we have made rigorous in \cite{DeNittis_Lein:sapt_photonic_crystals:2013}) crucially relies on the fact that we deal with first-order equations in time. 

If $V$ is a commuting or anticommuting linear or antilinear symmetry of $M$ that leaves $\Hil_{\mathrm{rel}}$ invariant (\ie $[V , P_{\mathrm{rel}}] = 0$), then $M_{\mathrm{eff}}$ necessarily inherits $V$ as at least an approximate symmetry, 
\begin{align*}
	V \, M_{\mathrm{eff}} \, P_{\mathrm{eff}} &\approx V \, M \, P_{\mathrm{eff}}
	= M \, P_{\mathrm{eff}} \, V 
	\approx M_{\mathrm{eff}} \, V \, P_{\mathrm{eff}} 
	. 
\end{align*}
Of course, we can also make this argument in reverse: any symmetry of the effective Maxwell operator is at least an approximate symmetry of the original, full Maxwell operator $M$. This immediately disqualifies effective “Hamiltonians” that possess an odd time-reversal symmetry due to the choice of material. 

\paragraph{Electric vs.\ magnetic vs.\ electromagnetic Chern numbers} 
\label{par:electric_vs_magnetic_vs_electromagnetic_chern_numbers}
The second-order formalism suggests that knowledge of the electric part of the Bloch functions suffices for all of our subsequent arguments. Indeed, many works (\eg \cite{Wang_et_al:edge_modes_photonic_crystal:2008}) compute the Chern numbers based only on the electric field. That is because to any given family of relevant frequency bands we can associated an \emph{electric} Bloch vector bundle $\mathcal{E}_{\BZ}^E$: just as described in Section~\ref{classification:Bloch_bundle} we glue together the electric subspaces $\Hil_{\mathrm{rel}}^E(k) = \mathrm{span} \bigl \{ \varphi_n^E(k) \bigr \}_{n \in \mathcal{I}}$ spanned by the \emph{electric} parts of the relevant Bloch functions. Such a bundle is then characterized up to continuous deformations by \emph{electric} (first and second) Chern numbers $C_j^E$, $j = 1 , 2$. 

Of course, we could base our arguments off of the magnetic field and obtain a \emph{magnetic} Bloch vector bundle $\mathcal{E}_{\BZ}^H$, characterized by \emph{magnetic} Chern numbers $C_j^H$, $j = 1 , 2$. That is all in addition to the \emph{electromagnetic} Bloch bundle $\mathcal{E}_{\BZ} = \mathcal{E}_{\BZ}^{EH}$ constructed in Section~\ref{classification:Bloch_bundle} and electromagnetic Chern numbers $C_j^{EH}$. \emph{A priori} there is no reason to believe that there necessarily exist simple relations between them, \eg we do \emph{not} know whether they sum up, $C_j^{EH} = C_j^E + C_j^H$, or whether $C_j^{EH} \neq 0$ implies $C_j^E \neq 0$ and $C_j^H \neq 0$ — or vice versa. Any such relations need to be verified at the very least by example, preferable through a mathematical proof. 

While the first-order formalism clearly singles out electromagnetic Chern number, from the vantage point of the second-order formalism it seems as if we have three choices. From a mathematical perspective, $C^E_j$ are collections of Chern numbers and as such integers, but it is no longer clear what physical significance they hold. Even though we do not know whether $C^E_1 \neq 0$ implies that also the electromagnetic Chern numbers $C^{EH}_1 \neq 0$ are non-zero, it is not unreasonable to use $C^E_1 \neq 0$ as a way to \emph{qualitatively} decide whether unidirectional boundary modes exist. However, we cannot expect this to translate to the \emph{quantitative} prediction Haldane has made (Conjecture~\ref{intro:conjecture:photonic_bulk_edge_correspondence}). 

\subsubsection{No analog of the Quantum Spin Hall Effect exists} 
\label{conclusion:comparison:QSHE}
There are several distinct effects which claim to be photonic analogs of the “Quantum Spin Hall Effect”, because the propagation direction of a boundary mode is locked to its spin or isospin degree of freedom (\eg \cite{Khanikaev_et_al:photonic_topological_insulators:2013,Chen_et_al:topological_photonic_crystals_crystallographic_symmetries:2015,Bliokh_Smirnova_Nori:spin_orbit_light:2015,Wu_Hu:topological_photonic_crystal_with_time_reversal_symmetry:2015}). The naming at the very least suggests that topology is at the heart of this spin-momentum locking. 

In the context of quantum solid state physics, the Quantum Spin Hall Effect \cite{Kane_Mele:Z2_ordering_spin_quantum_Hall_effect:2005} has a very precise meaning: the system is of \emph{class~AII}, meaning it possesses an \emph{odd} time-reversal symmetry, and the presence of this symmetry is the immediate cause for the spin-momentum locking. However, as we have discovered here electromagnetic media do not support \emph{odd} time-reversal symmetries — only even ones are admissible (\cf Proposition~\ref{symmetries:prop:admissible_symmetries}). Thus, in a topological sense, these phenomena are \emph{not} photonic analogs of the Quantum Spin Hall Effect. Nevertheless, of the four works mentioned, one of them \cite{Wu_Hu:topological_photonic_crystal_with_time_reversal_symmetry:2015} actually \emph{does} describe a topological effect.

\paragraph{Bliokh et al's “Quantum Spin Hall Effect for light”} 
\label{par:bliokh_et_al_s_quantum_spin_hall_effect}
Interfaces between electromagnetic media such as between a dielectric and a metal (so that the sign of $\eps$ flips) can support surface modes which are localized to the vicinity of the surface layer. Such surface waves are also known as surface plasmon-polaritons. 
For simplicity, let us focus on interfaces between two \emph{homogeneous} media, as that situation admits an explicit analytical solution.
Bliokh, Smirnova and Nori \cite{Bliokh_Smirnova_Nori:spin_orbit_light:2015} gave a very concise and elegant explanation for the following effect: if these surface waves are excited with circularly polarized light, then the propagation direction of the surface wave is locked to the handedness of the incoming light. To be clear, despite calling this an analog of the “Quantum Spin Hall effect” Bliokh et al did \emph{not} claim that this effect is of topological origin. Instead, spin-momentum locking here is due to the \emph{transversality constraint} and the fact that \emph{surface modes come in a single, fixed polarization} — they are necessarily transverse magnetic waves. 

The localization of the surface wave to the interface means that both, the local wave vector $k = (k_x , \ii k_z)$ and the polarization vector $\Psi_{\mathrm{surf}}(k) = \bigl ( \psi_{\mathrm{surf}}^E(k) \, , \, \psi_{\mathrm{surf}}^H(k) \bigr )$ are necessarily complex (\cf \cite[equation~(4)]{Bliokh_Smirnova_Nori:spin_orbit_light:2015}), for otherwise the transversality constraint 
\begin{align*}
	k \cdot \psi^E_{\mathrm{surf}}(k) = 0 = k \cdot \psi^H_{\mathrm{surf}}(k) 
\end{align*}
could not be satisfied. The fact that these polarization vectors are complex forces the electric field into a rotation in the plane of propagation, which, in turn, gives rise to a spin angular momentum \emph{transverse} to the plane. And the sense of rotation, \ie the sign of the transverse spin is locked to the propagation direction of the surface wave (\cf \cite[Figure~3A]{Bliokh_Smirnova_Nori:spin_orbit_light:2015}); time-reversal symmetry, which flips spin and and the in-plane momentum $k_x$, relates these two counterpropagating waves. 

Because the incoming circularly polarized light can only excite surface modes whose sense of rotation matches its own, the sense of rotation dictates which of the two counter-propagating surface modes will be excited — and therefore the propagation direction. This robust mechanism, which is responsible for spin-momentum locking, is at the heart of a broad range of phenomena \cite{Bliokh_et_al:spin_orbit_light:2015,Hatano_et_al:photo_rectification_effect:2009,Kurosawa_et_al:photo_rectification_effect:2012}. However, this is not a topological effect, in particular not an analog of the Quantum Spin Hall Effect in the topological sense. 

\paragraph{Time-reversal symmetric media with isospin symmetry} 
\label{par:time_reversal_symmetric_media_with_isospin_symmetry}
Manufacturing electromagnetic media for which time-reversal symmetry is broken at optical frequencies is difficult. To circumvent that difficulty several researchers had had the idea to \emph{add} an isospin symmetry to a medium \emph{with} time-reversal symmetry. One way to do that is by introducing a crystallographic symmetry: Wu and Hu \cite{Wu_Hu:topological_photonic_crystal_with_time_reversal_symmetry:2015} proposed to arrange identical rods made of a dielectric in a hexagonal lattice and sandwiched between two metal plates. To be specific, the electric permittivity $\eps(x) = \eps_{\mathrm{d}}(x) \; \id$ is a real scalar, the magnetic permeability $\mu(x) = \mu_0$ equals the vacuum value and the bianisotropic tensor $\chi = 0$ is absent; were it not for the additional isospin symmetry, this material would be of class~AI. Instead, it possesses \emph{three} symmetries, namely the \emph{even} time-reversal symmetry $T_3$, the (unitary, commuting) isospin symmetry as well as their product. 

Provided the parameters are chosen correctly, such photonic crystals have a photonic band gap for TM modes and there are edge modes at the boundary (\cf \cite[Figure~3]{Wu_Hu:topological_photonic_crystal_with_time_reversal_symmetry:2015}). Due to the presence of the time-reversal symmetry $T_3$ these boundary modes come in \emph{pairs} that are mirror symmetric with respect to reflection $k_x \mapsto -k_x$; here, $k_x$ is the Bloch momentum associated to the remaining periodic direction parallel to the edge. Moreover, each of these boundary modes have a \emph{definite pseudospin}. This leads to a locking between pseudospin and propagation direction: according to the band diagram \cite[Figure~5]{Wu_Hu:topological_photonic_crystal_with_time_reversal_symmetry:2015} the boundary modes associated to isospin $\uparrow$ has a strictly positive group velocity ($\partial_x \omega_{\mathrm{edge},\uparrow} > 0$) while its symmetric partner $\downarrow$ has strictly negative group velocity ($\partial_x \omega_{\mathrm{edge},\downarrow} < 0$). Moreover, they propose that this is described by an effective Hamiltonian with an \emph{odd} time-reversal symmetry. Therefore, Wu and Hu incorrectly argue that this is an analogue to the Quantum \emph{Spin} Hall Effect and the system possesses a $\Z_2$-valued invariant. 

It is worthwhile to outline their argument in order to trace the mistakes in their analysis: first of all, they employ the second-order formalism which is not suited for the symmetry classification (\cf \cite[Section~3]{DeNittis_Lein:symmetries_Maxwell:2014}) as it becomes impossible to distinguish between commuting and anticommuting symmetries of the auxiliary Maxwell operator $\Maux_+ = \Maux_-$. Wu and Hu erroneously identify complex conjugation as being of time-reversal type,  
\begin{align*}
	C \, \bigl ( \Maux_+ \bigr )^2 \, C = (-1)^2 \, \bigl ( \Maux_+ \bigr )^2 = \bigl ( \Maux_+ \bigr )^2 
\end{align*}
whereas $C$ actually \emph{anti}commutes with $\Maux_+$. In fact, because the real-valuedness of electromagnetic fields is one of the tenets of electromagnetism, the equations describing electromagnetic waves in media can \emph{never} break complex conjugation symmetry. Therefore, this unbreakable symmetry is not relevant for the topological classification. 

Not only this point, but also the next item illustrates why the second-order formalism is unsuitable for making quantum-wave analogies rigorous: The authors then note the (avoided) conical intersection at $k = 0$ and propose an effective operator for those which possesses an odd time-reversal symmetry. Because this effective $4$-band operator \cite[Supplementary Material, equations~(S28)–(S29)]{Wu_Hu:topological_photonic_crystal_with_time_reversal_symmetry:2015} supposedly approximates $\bigl ( \Maux_+ \bigr )^2$, it is quadratic instead of linear in $k$. We caution against making such \emph{ad hoc} arguments without making any reference to the dynamical problem – the reasoning to replace a Hamiltonian with an effective one \emph{crucially} relies on the equations of motion being \emph{first order in time} (\eg via a Duhamel argument akin to equation~\eqref{conclusion:eqn:Duhamel_argument}). Without having established a direct link between the “effective Hamiltonian” and Maxwell's equations, the presence of an odd time-reversal symmetry for the “effective Hamiltonian” does not imply that the original equations, Maxwell's equations, sport such a symmetry as well. 

Then the frequency bands and their Chern numbers are computed; two of these bands have Chern numbers $0$ while the other two have Chern numbers $\pm 1$. The $\Z_2$-invariant is the Chern number mod $2$. Note that due to the presence of the even time-reversal symmetry \emph{Chern numbers come in pairs of equal magnitude and opposite sign.} 
\medskip

\noindent
We can give a simple and systematic explanation of Wu's and Hu's finding, the locking of spin and momentum in the boundary modes, by applying the topological classification tools of Section~\ref{classification}. In particular, the presence of topologically protected boundary modes does not contradict our main classification result (Theorem~\ref{intro:thm:bulk_classification}). 

The system Wu and Hu studied exemplifies why assuming the absence of additional unitary, commuting symmetries (Assumption~\ref{intro:assumption:no_additional_symmetries}) is absolutely crucial for the physics and not a mathematical footnote. The existence of a unitary, commuting (isospin) symmetry means we can decompose the Maxwell operator 
\begin{align*}
	M(k) = M_{\uparrow}(k) \oplus M_{\downarrow}(k)
	= \left (
	\begin{matrix}
		M_{\uparrow}(k) & 0 \\
		0 & M_{\downarrow}(k) \\
	\end{matrix}
	\right )
\end{align*}
into block operators $M_{\uparrow/\downarrow}$ that act on the isospin $\uparrow$/$\downarrow$ subspaces. Now two things may happen: \emph{either the time-reversal symmetry is block-diagonal or it is not.} The block operators for dual symmetric media retain a time-reversal symmetry whereas for the system that Wu and Hu consider, it turns out to be broken as we will explain below. Put another way, we are dealing with a “2 $\times$ class~A” system (one for each isospin eigenstate). The topology of such a system is completely determined by the Chern numbers of the isospin-$\uparrow$ bands; again, thanks to the even time-reversal symmetry $T_3$ the Chern numbers of the isospin-$\downarrow$ bands are necessarily equal in magnitude but have opposite sign compared to their symmetric isospin-$\uparrow$ partners. 

The simplest way to see this is by observing that akin to \cite[Supplementary Material, equation~(S11)]{Wu_Hu:topological_photonic_crystal_with_time_reversal_symmetry:2015} time-reversal symmetry \emph{flips} isospins, and is therefore completely block-\emph{off}diagonal in the isospin basis. Retracing these arguments on the level of Maxwell's equations with the correct symmetries is straightforward but lengthy — and fortunately for us unnecessary. 

That is because breaking of time-reversal symmetry can be deduced solely from the band picture obtained by Wu and Hu (\cf \cite[Figure~5]{Wu_Hu:topological_photonic_crystal_with_time_reversal_symmetry:2015}): if the time-reversal symmetry were block-diagonal, then $M_{\uparrow/\downarrow}$ would inherit a time-reversal symmetry. Therefore, isospin-$\uparrow$ edge bands would necessarily come in pairs, and that would mean the edge modes had to be two-fold spin degenerate. However, according to the \cite[Figure~5]{Wu_Hu:topological_photonic_crystal_with_time_reversal_symmetry:2015} the edge modes are non-degenerate and have a definite isospin — $M_{\uparrow/\downarrow}$ cannot possess a time-reversal symmetry. 

In summary, $M_{\uparrow/\downarrow}$ are operators of \emph{class~A}, the same topological class as gyrotropic media or quantum systems exhibiting the Quantum Hall Effect, and the topological invariants are the usual $\Z$-valued Chern numbers, rather than a $\Z_2$-valued invariant. Given the topological classification, the spin-momentum locking \emph{cannot} be seen as a topological analog of the Quantum \emph{Spin} Hall Effect; instead, it is analogous to the usual Quantum Hall Effect. These Chern numbers come in isospin pairs: due to the time-reversal symmetry of the total system, isospin-$\downarrow$ Chern numbers are equal in magnitude but have opposite sign compared to their symmetric isospin-$\uparrow$ partner — after all, they need to sum to $0$. This correct classification explains why edge modes of given isospin are unidirectional and are afforded \emph{topological protection} — provided that the perturbation preserves the isospin symmetry. However, generic perturbations, \ie those that break the honeycomb symmetry, will mix isospin states and backscattering may occur. That is why topological effects which are due to crystallographic symmetries are less robust than those which only depend on symmetries of the medium; however, in light of the difficult of fabricating media which are gyrotropic and lossless in the optical or infrared, this may still be a worthwhile tradeoff in practice. 

\paragraph{Certain bianisotropic media} 
\label{par:dual_symmetric_non_gyrotropic_media}
A different approach is taken by Khanikaev et al \cite{Khanikaev_et_al:photonic_topological_insulators:2013} as well as Chen et al \cite{Chen_et_al:topological_photonic_crystals_crystallographic_symmetries:2015} who propose to realize an isospin degree of freedom due to a symmetry of the material weights 
\begin{align*}
	W = \left (
	\begin{matrix}
		\eps & \chi \\
		\chi & \eps \\
	\end{matrix}
	\right ) = \id \otimes \eps + \sigma_1 \otimes \chi 
	\neq \overline{W} = \id \otimes \eps - \sigma_1 \otimes \chi 
\end{align*}
where $\eps = \mathrm{diag} \bigl ( \eps_{\perp},\eps_{\perp},\eps_z \bigr )$ is purely diagonal and the bianistropic tensor 
\begin{align*}
	\chi = \left ( 
	\begin{matrix}
		0 & + \ii \chi_{xy} & 0 \\
		- \ii \chi_{xy} & 0 & 0 \\
		0 & 0 & 0 \\
	\end{matrix}
	\right )
\end{align*}
is purely imaginary and offdiagonal. According to the nomenclature introduced here, the medium described by $W$ possesses a single even time-reversal symmetry, $T_3$, and therefore belongs to the non-gyrotropic class of media (class~AI, \cf the table in Theorem~\ref{intro:thm:classification_media}). Absent any crystallographic symmetries, our analysis shows there are no topological effects — in direct contradiction to the predictions made in \cite{Khanikaev_et_al:photonic_topological_insulators:2013,Chen_et_al:topological_photonic_crystals_crystallographic_symmetries:2015}. We will explain where their mistake lies. 

Once we express $W = \id \otimes \eps + \sigma_1 \otimes \chi$ in terms of Pauli matrices, we immediately see that the weights commute with the operator $J_1 = \sigma_1 \otimes \id$. The authors of \cite{Khanikaev_et_al:photonic_topological_insulators:2013,Chen_et_al:topological_photonic_crystals_crystallographic_symmetries:2015} exploited the fact that $W$ and $J_1$ can be diagonalized simultaneously proceeded to split Maxwell's equations into separate “spin-up” and “spin-down” equations. Time-reversal symmetry $T_3$ maps between spin-up and spin-down states. The authors then show the existence of boundary modes; especially \cite{Chen_et_al:topological_photonic_crystals_crystallographic_symmetries:2015} seems to make a very persuasive argument as they even verify the bulk-edge correspondence quantitatively. 

Unfortunately, \emph{the analyses in \cite{Khanikaev_et_al:photonic_topological_insulators:2013,Chen_et_al:topological_photonic_crystals_crystallographic_symmetries:2015} start out with unphysical equations} — “spin” eigenmodes $\bigl ( \mathbf{E} + \mathbf{H} \, , \,  \mathbf{E} - \mathbf{H} \bigr )$ are \emph{not transversal}. There are two main reasons for this: first of all, $J_1$ \emph{anti}commutes with the free Maxwell operator $\Rot = \left ( 
\begin{smallmatrix}
	0 & + \ii \nabla^{\times} \\
	- \ii \nabla^{\times} & 0 \\
\end{smallmatrix}
\right )$ and therefore also \emph{anti}commutes with the product $\Maux = W^{-1} \, \Rot$. That means $J_1$ necessarily \emph{maps positive onto non-positive states.} Secondly, the authors did not take into account that negative frequency states are governed by a different set of Maxwell's equations that involve the \emph{complex conjugate} weights $\overline{W} \neq W$. 

As a consequence of $J_1 \, \Maux \, J_1 = - \Maux$, there exist no positive frequency solutions to the spin eigenvalue equation $J_1 \Psi = \pm \Psi$. In fact, $J_1$ is not even well-defined as an operator on the non-negative frequency subspace $\Hil \longrightarrow \Hil$. The same arguments apply to the $\omega \leq 0$ subspace. Consequently, spin eigenstates necessarily contain linear combinations of positive \emph{and} negative frequency states. 

The last fact makes verifying the transversality of spin eigenstates more difficult as we have two distinct transversality conditions~\eqref{Schroedinger:eqn:Maxwell_equations:constraint} that mix — one for positive frequency waves with weights $W_+ = W$ and one for negative frequencies with the complex conjugate weights $W_- = \overline{W} \neq W$. The fact that the projections $\Pi_{\pm} = \tfrac{1}{2} \bigl ( \id \pm J_1 \bigr )$ onto the eigenspaces are explicit means that we need to check whether for an arbitrary positive frequency state $\Psi$ the complex electromagnetic field $J_1 \Psi$ (composed solely of non-positive frequencies!) is transversal. However, a direct computation shows 
\begin{align*}
	\overline{W} \, J_1 &= - \id \otimes \chi + \sigma_1 \otimes \eps 
	\neq W = \id \otimes \eps + \sigma_1 \otimes \chi 
	, 
\end{align*}
and $J_1 \Psi$ violates the transversality condition, 
\begin{align*}
	\Div \, \overline{W} \, J_1 \Psi \neq \Div \, W \, \Psi = 0 
	. 
\end{align*}
In fact, replacing $\overline{W}$ with $W$ in the above computation shows that $J_1 \Psi$ also does not satisfy the \emph{positive} frequency transversality condition either. Hence, spin eigenstates $\bigl ( \mathbf{E} + \mathbf{H} \, , \,  \mathbf{E} - \mathbf{H} \bigr )$ cannot be transversal, they violate one of Maxwell's equations~\eqref{Schroedinger:eqn:Maxwell_equations:constraint}. 

\subsection{Possible directions for future work} 
\label{conclusion:future}
The premise of the article was to distinguish between material symmetries, which have been studied here, and crystallographic symmetries. Because electromagnetic and media for other classical waves can be fabricated to exact specifications, there is a strong interest in understanding the role that crystallographic symmetries play. Our analysis of \cite{Wu_Hu:topological_photonic_crystal_with_time_reversal_symmetry:2015} shows the power our approach holds — we were straightfowardly able to link the presence of a $C_6$-symmetry to a topological effect without making reference to an “effective Hamiltonian” whose link to the original equations is tenuous. What is more, rewriting Maxwell's equations in the form of a Schrödinger equation (that is first order in time!) opens the door to the rich library of results from the condensed matter community (see \eg the references in \cite{Hasan_Kane:topological_insulators:2010,Chiu_Teo_Schnyder_Ryu:classification_topological_insulators:2016,Prodan_Schulz_Baldes:complex_topological_insulators:2016}). 

Secondly, we still owe an answer to the first of the two questions from the introduction — a derivation of Haldane's Phontonic Bulk-Boundary Conjecture. That requires us to tackle the problems outlined in Section~\ref{conclusion:comparison:QHE}. 

Then we intend to look beyond electromagnetism and broaden our considerations to include other classical waves. Indeed, analogs of the Quantum Hall Effect have been realized with other classical waves, including certain acoustic waves \cite{Fleury_et_al:breaking_TR_acoustic_waves:2014,Safavi-Naeini_et_al:2d_phonoic_photonic_band_gap_crystal_cavity:2014,Peano_Brendel_Schmidt_Marquardt:topological_phases_sound_light:2015,Chen_Zhao_Mei_Wu:acoustic_frequency_filter_topological_phononic_crystals:2017} and coupled pendula \cite{Suesstrunk_Huber:mechanical_topological_insulator:2015,Suesstrunk_Huber:classification_mechanical_metamaterials:2016}, and we would like to understand to what extent these phenomenological similarities are founded on similar mathematics. Indeed, quite a few wave equations share the same essential structure that characterize Maxwell's equations: they are first-order in time, feature a product structure and the physical fields are real. For such wave equations the construction outlined in \cite[Section~6]{DeNittis_Lein:Schroedinger_formalism_classical_waves:2017} yields a Schrödinger formalism. Just like with electromagnetism this will allow us and others to reach into the rich toolbox of techniques initially developed for quantum systems and apply them to classical waves; the topological classification of media these waves propagate in would be but one example. 

A mathematically and physically very intriguing question concerns the generalization of our Schrödinger formalism and mathematical classification to metals and metamaterials which violate the essential assumption that the material weights be positive. One mathematical consequence is that $\scpro{\, \cdot \,}{\, \cdot \,}_W$ fails to be a scalar product; similar to measuring distances with respect to the Lorentzian metric, $\scpro{\Psi}{\Psi}_W = 0$ can vanish even if $\Psi \neq 0$. Spaces with such indeterminate inner products are called \emph{Krein spaces} \cite{Azizov_Iokhvidov:Krein_spaces:1989}, and their mathematical properties are very different from those of Hilbert spaces. Unfortunately, Krein-selfadjointness (Krein-hermitianness) is a much weaker notion, \eg Krein-selfadjoint operators do not necessarily possess a functional calculus, and because of that it is \emph{a priori} not even clear how to implement the restriction to positive frequencies. Exploring a topological classification and more advanced topics such as bulk-boundary correspondences would be mathematically very intriguing while at the same time having clear physical applications. Perhaps these give rise to topological effects with no quantum analog \cite{Bracci_Morchio_Strocchi:Wigners_theorem_Krein_spaces:1975}. One “drosophila system” we have our eyes on are spin waves \cite{Shindou_et_al:chiral_magnonic_edge_modes:2013} where $W = \sigma_3 \otimes \id$; they lack many of the technical complications yet still force us to work with an indeterminate inner product. 

\printbibliography

\end{document}